\documentclass[11pt]{amsart}

\usepackage[margin=1in]{geometry}
\usepackage[T1]{fontenc}
\usepackage{lmodern}
\usepackage{amsmath,amssymb,amsthm,mathtools,bm,mathrsfs}
\usepackage{enumitem,microtype}
\usepackage{hyperref}
\usepackage[nameinlink,capitalize]{cleveref}

\newtheorem{theorem}{Theorem}[section]

\newtheorem{proposition}[theorem]{Proposition}
\newtheorem{corollary}[theorem]{Corollary}

\newtheorem{remark}[theorem]{Remark}
\newcommand{\commentout}[1]{}
\newcommand{\C}{\mathbb C}
\newcommand{\R}{\mathbb R}
\newcommand{\Sph}{\mathbb S}
\newcommand{\SO}{\mathrm{SO}}
\newcommand{\cH}{\mathcal H}
\newcommand{\cU}{\mathcal U}
\newcommand{\cT}{\mathcal T}

\newcommand{\cB}{\mathcal B}
\newcommand{\cC}{\mathcal C}
\newcommand{\cJ}{\mathcal J}
\newcommand{\cQ}{\mathcal Q}
\newcommand{\cL}{\mathcal L}
\newcommand{\cZ}{\mathcal Z}
\newcommand{\Range}{\operatorname{Range}}
\newcommand{\diag}{\operatorname{diag}}
\newcommand{\rank}{\operatorname{rank}}

\newcommand{\argmin}{\operatorname*{argmin}}
\newcommand{\dd}{\,d}

\newcommand{\bomega}{{\bm\omega}}

\newcommand{\bgamma}{{\bm\gamma}}
\newcommand{\bx}{{\bm x}}
\newcommand{\by}{{\bm y}}
\newcommand{\bz}{{\bm z}}
\newcommand{\bh}{{\bm h}}
\newcommand{\be}{{\bm e}}
\newcommand{\ba}{{\bm a}}

\newcommand{\bb}{{\bm b}}
\newcommand{\bc}{{\bm c}}
\newcommand{\bv}{{\bm v}}
\newcommand{\bu}{{\bm u}}
\newcommand{\bd}{{\bm d}}
\newcommand{\bw}{{\bm w}}
\newcommand{\bPi}{{P}}%{\boldsymbol\Pi}
\newcommand{\eps}{\varepsilon}
\usepackage{booktabs}
\usepackage{xurl}
\newcommand{\norm}[1]{\left\lVert#1\right\rVert}

\newcommand{\cond}{\operatorname{cond}_2}
\newcommand{\sinc}{\operatorname{sinc}}
\newcommand{\round}{\operatorname{round}}
\newcommand{\file}[1]{\nolinkurl{#1}}
\hypersetup{
 colorlinks=true,linkcolor=blue,citecolor=blue,urlcolor=blue,
 pdfauthor={Albert Fannjiang},
 pdftitle={Spherical ESPRIT by Paired Small Circles and Infinitesimal Rotations},
 pdfsubject={Finite-shift and generator ESPRIT, identifiability, perturbations, and MUSIC refinement},
 pdfkeywords={spherical ESPRIT, paired small circles, rotation generators, phase unwrapping, spherical MUSIC}
}
\title[Spherical ESPRIT by finite shifts and generators]{Spherical ESPRIT by Paired Small Circles\\
and Infinitesimal Rotations}
\author{Albert Fannjiang}
\date{September 24, 2026}
\begin{document}
\raggedbottom
\begin{abstract}
We study the recovery of a finite three-dimensional point cloud from
its spherical Fourier signal subspace at a fixed wavenumber. Two ESPRIT
constructions are analyzed. The first uses finite shifts between paired
small circles; the second uses angular derivatives on the whole sphere.
Both recover the same commuting coordinate matrices, either through
their finite exponentials or directly from a coupled linear system.
We prove that the continuous generator system is injective for every
distinct point cloud and that one guard harmonic degree preserves the
retained equations exactly. For $s$ targets, the cutoff $K\ge s$ gives
generic finite identifiability, but not uniform conditioning. We also
derive a finite-cutoff perturbation bound.

The accuracy of the coordinate estimates is a separate issue. A
one-target perturbation calculation shows that the original generator
estimate can retain a nonzero bias as the wavenumber increases. To
address this bias, we use the generator eigenbasis to separate
approximate atoms and recover their locations from amplitude-normalized
phases. Under pointwise atom dominance, the resulting error is
$O(\kappa^{-1})$; for a dominant two-atom mixture at fixed separation,
it is $O(\kappa^{-2})$. Experiments on structured point clouds with
smooth perturbations compare the two constructions and their use as
initializers for a common MUSIC refinement.
\end{abstract}
\maketitle

\tableofcontents

\section{Introduction}

\label{sec:introduction}

Inverse scattering separates naturally into the estimation of a signal
subspace from measurements and the recovery of the scatterer locations
from that subspace. This paper concerns the second problem. We seek
algebraic reconstruction methods for a finite point cloud from its
spherical Fourier signal subspace at one wavenumber.

Let $X=\{\bx_1,\ldots,\bx_s\}\subset B_1(0)\subset\R^3$ consist of
distinct points, and let $\kappa>0$. Define the spherical Fourier atoms
and their synthesis operator by
\begin{align}
 \varphi_{\bx}(\bomega)&=e^{i\kappa\bomega\cdot\bx},
       &&\bomega\in\Sph^2,\label{eq:atom}\\
 \Phi_X\bc&=\sum_{j=1}^s c_j\varphi_{\bx_j},
       &&\Phi_X:\C^s\longrightarrow
       \cH=L^2\!\left(\Sph^2,\frac{\dd\bomega}{4\pi}\right).
       \label{eq:Vandermonde}
\end{align}
The observation variable $\bomega$ lies on the unit sphere, but the unknown
locations $\bx_j$ range throughout the ball. Thus both direction and
radial position are to be recovered. A known support radius $R_0$ can
be reduced to this setting by replacing $\bx$ with $\bx/R_0$ and
$\kappa$ with $\kappa R_0$. For $s\ge2$, let
$\delta_X=\min_{j\ne k}|\bx_j-\bx_k|$ denote the minimum separation.
The signal subspace is $U_X=\Range\Phi_X$, and any basis of it has
the form
\begin{equation}
 Q=\Phi_XR,\qquad R\in GL(s,\C),
 \label{eq:Q-basis}
\end{equation}
and the associated coordinate representation is
\begin{equation}
 D_r(X)=\diag(x_{1,r},\ldots,x_{s,r}),\qquad
 \Psi_r=R^{-1}D_r(X)R,\qquad r=1,2,3.
 \label{eq:coordinate-matrices}
\end{equation}
The matrices $\Psi_r$ commute and have joint spectrum $X$. Recovery
therefore reduces to constructing these matrices, or their finite
exponentials, from the given basis $Q$.

The subspace arises from snapshots $\mathcal Y=\Phi_XC$ when the
coefficient matrix $C\in\C^{s\times N_{\rm snap}}$ has row rank $s$.
In the point-scatterer model, illumination and multiple scattering enter
through the effective coefficients, while the observation-direction
dependence remains in $\Phi_X$. The coefficients must be the same on
all aperture restrictions used in the reconstruction. We assume that
an exact or estimated signal subspace is given; its stable extraction
from physical measurements is a separate problem.

Classical ESPRIT obtains unknown phases by comparing compatible shifted
restrictions of a signal basis \cite{roykailath1989}. On a spherical
aperture, a nonzero translation cannot map a two-dimensional open patch
back onto the sphere. It can, however, map one small circle onto another.
For a Fourier displacement $q\bv$, with $|\bv|=1$ and $0<q<2\kappa$,
the paired traces satisfy
\begin{equation}
 Q_+=Q_-F_{q,\bv},\qquad
 F_{q,\bv}=R^{-1}\diag(e^{iq\bv\cdot\bx_j})R
          =\exp\!\left(iq\sum_rv_r\Psi_r\right).
 \label{eq:intro-paired-identity}
\end{equation}
If the restrictions have full column rank, a common eigenbasis pairs
the phases and determines the projected coordinates, subject to phase
ambiguity. A larger displacement divides the phase error by a larger
$q$, but requires well-conditioned restrictions and correct phase
continuation.

The infinitesimal construction instead uses
$\mathcal L_i=-i(\bomega\times\nabla_{\Sph^2})_i$ and
$(M_qf)(\bomega)=\omega_qf(\bomega)$ to obtain
\begin{equation}
 \mathcal L_iQ=\kappa\sum_{q,r=1}^3
          \epsilon_{iqr}M_qQ\Psi_r,\qquad i=1,2,3.
 \label{eq:intro-esprit-intertwining}
\end{equation}
These equations determine all three coordinate matrices through one
coupled whole-sphere system, without phase unwrapping. Stability depends
on the angular derivatives, the conditioning of the coupled system,
and the subsequent spectral calculation. The rotation generators act
on aperture functions; the commuting coordinate matrices act on signal
coefficients.

Although the two constructions agree on the exact coordinate
representation, they need not give the same estimates from a perturbed
subspace. Restriction and differentiation act differently on the error.
There is a further distinction between separating approximate atoms and
estimating their coordinates: a common eigenbasis may still separate
the atoms when its diagonal coordinate estimates are biased. This
observation leads to a phase gradient method that uses the generator eigenbasis
for demixing and the phases of the demixed functions for localization.
We compare all three estimates before and after a common MUSIC
refinement \cite{schmidt1986,fannjiangnguyen2026}.
Table~\ref{tab:method-comparison} summarizes the two ESPRIT constructions.

\begin{table}[htbp]
\centering\small
\begin{tabular}{@{}p{0.22\textwidth}p{0.35\textwidth}p{0.35\textwidth}@{}}
\toprule
Feature & Paired small circles & Whole-sphere generators\\
\midrule
Exact action & Finite translation phases & Infinitesimal rotation identities\\
Aperture geometry & Two parallel small circles & Entire observation sphere\\
Coefficient algebra & Commuting finite shifts & Commuting Cartesian matrices\\
Continuous rank & Distinct projected nodes & Distinct nodes at full aperture\\
Finite realization & Circle samples or harmonic restrictions & Guarded spherical harmonics\\
Perturbation sensitivity & Restriction and phase continuation & Differentiation and coupled block solve\\
Coordinate estimate & Unwrapped phases divided by displacement & {\raggedright Two-sided diagonal estimate or demixed phase integral\par}\\
\bottomrule
\end{tabular}
\caption{The two aperture realizations recover the same Cartesian
coefficient algebra but have different rank and conditioning mechanisms.}
\label{tab:method-comparison}
\end{table}

\section{Related literature and contributions}

Roy and Kailath \cite{roykailath1989} formulate the basic ESPRIT
relation between two compatible restrictions of a signal basis. The
resulting coefficient matrix is similar to a diagonal matrix of unknown
phase factors. Multidimensional recovery requires a common ordering
of the coordinate estimates. Haardt and Nossek \cite{haardtnossek1998}
obtain this ordering by simultaneous Schur decomposition, while Ehler,
Kunis, Peter, and Richter \cite{ehler2018} use a random linear
combination in a multivariate matrix-pencil method.

The common spectral calculation is generally nonunitary. Balda et al.\
\cite{balda2017} analyze perturbations of joint eigenvalue
decompositions, and He, Kressner, and Plestenjak
\cite{hekressnerplestenjak2025} study two-sided coordinate estimates
from random pencils, including nonnormal matrices. These methods
provide the spectral pairing used here. Their conditioning is distinct
from that of the aperture equations that determine the matrices.

Andersson and Carlsson \cite{anderssoncarlsson2018} extend
multidimensional ESPRIT to general domains by using general-domain
Hankel operators and compatible shifted restrictions. A fixed-frequency
spherical shell imposes an additional geometric constraint: a nonzero
translation cannot preserve an open surface patch. Our finite-shift
construction uses paired small circles, while the differential
construction uses intrinsic derivatives on the full shell. Neither
requires Fourier data away from the prescribed sphere.

Stability and resolution have been studied for classical ESPRIT.
Fannjiang \cite{fannjiang2016esprit} treats single-snapshot recovery,
and Li, Liao, and Fannjiang \cite{liliaofannjiang2020} obtain
super-resolution estimates in terms of Vandermonde singular values
and clustered-node geometry. In the present setting, conditioning of
the scalar synthesis operator is only one requirement. The circle
restriction or the coupled angular system must also be well conditioned.

Spherical-harmonic ESPRIT has a substantial direction-of-arrival
literature. Goossens and Rogier \cite{goossensrogier2009} develop
unitary spherical ESPRIT for two-angle estimation. Jo and Choi
\cite{jochoi2019} combine three harmonic recurrences to avoid
coordinate singularities and pairing ambiguities. Herzog and Habets
\cite{herzoghabets2019} likewise recover direction vectors through
three recurrences and joint diagonalization. In these formulations
the unknowns are unit directions, and known array-dependent radial
responses are compensated on the usable harmonic modes.

Radial recovery also has precedents. Herzog and Habets
\cite{herzoghabets2022} estimate source distance without a spatial
search by using spherical-Hankel recurrences. Their spherical-wave
model differs from the Fourier-subspace model considered here, whose
harmonic coefficients contain the unknown factors
$j_\ell(\kappa|\bx_j|)$. Both of our constructions retain this
target-dependent radial information: the generator equations recover
Cartesian coordinates, and the paired shifts recover their projections.
Choi, Zotter, Jo, and Yoo \cite{choietal2022} obtain three-dimensional
localization from multiple spherical arrays and cross-array
covariances. Here the input is instead a spherical Fourier signal
subspace at a single wavenumber. The results below concern the
aperture identities, identifiability, and perturbation properties of
the two reconstructions from this subspace.

\subsection{Contributions of this paper}
\label{sec:esprit-contributions}

The main results concern exact recovery, finite realization, and the
effect of perturbations on coordinate estimation.

\begin{enumerate}[leftmargin=*,label=\textup{(\roman*)}]
\item \textbf{Exact recovery.}
Paired restrictions recover the targets under full restriction rank and
unaliased phases; the continuous generator block is injective for every
distinct cloud at every $\kappa>0$
(\cref{thm:paired-exact-recovery,thm:exact-whole-sphere}).

\item \textbf{Finite-shift and differential relations.}
A tangential Hodge isometry identifies the gradient and angular-generator
least-squares formulations. A synchronized small-circle limit links the
paired identity to the local generator identity, with the rotation axis
varying along the equator
(\cref{thm:Hodge-equivalence,thm:great-circle-limit}).

\item \textbf{Finite realization and identifiability.}
One guard harmonic degree preserves the retained generator identity
exactly, and $K\ge s$ suffices for finite identifiability for generic
$s$-point configurations
(\cref{thm:guard-shell,thm:rank-generic-cutoff}).
Proposition~\ref{prop:finite-perturbation} bounds the coordinate-matrix
error at a fixed cutoff.

\item \textbf{Perturbation mechanisms and phase-gradient recovery.}
A one-target perturbation calculation gives a persistent bias for the
original generator estimate and a decreasing coordinate error for
reliably continued long shifts (\cref{app:scalar-kernel}).
After demixing, amplitude-normalized phase recovery has an
$O(\kappa^{-1})$ error bound under pointwise atom dominance. For a
dominant two-atom mixture at fixed separation, the bound improves to
$O(\kappa^{-2})$
(\cref{thm:hybrid-pointwise-dominance}, \eqref{eq:phase-two-atom-bound}).

\item \textbf{Controlled comparisons.}
Experiments on ten structured 125-target clouds over
$\kappa=10$--$1280$ exhibit regimes favoring each initializer under a
common empirical MUSIC iteration. A phase-gradient readout holds the generator
eigenbasis fixed, and a separate common-input benchmark controls the data
representation (\cref{sec:matched-regimes}).
\end{enumerate}

\subsection{Organization.}
Sections~\ref{sec:paired-method} and \ref{sec:whole-sphere-method}
give the exact recovery results.
Section~\ref{sec:methods-relation} relates the finite-shift and
differential identities, and Section~\ref{sec:finite-realization}
establishes the guarded realization, generic rank, and finite-cutoff
perturbation bound. The algorithms and numerical experiments are
given in Sections~\ref{sec:algorithms-perturbations}
and \ref{sec:matched-regimes}. We conclude in
Section~\ref{sec:outlook}. The appendices contain the aperture algebra,
the generic-rank proof, and the bias and phase-recovery analysis.

\section{Paired-small-circle ESPRIT}

\label{sec:paired-method}

A fixed chord pairs two small circles on the spherical aperture.
The corresponding restrictions satisfy a finite-shift identity, from
which the target coordinates can be recovered by a common spectral
calculation. We first derive this identity and then determine when
the restrictions have full column rank.

Throughout this section, $\bh$ is a dimensionless aperture chord and
$\kappa\bh$ is its Fourier displacement. For a unit vector $\bv$, set
\begin{equation}
\label{eq:equator-fiber}
 E_{\bv}=\{\bomega\in\Sph^2:\bomega\cdot\bv=0\}.
\end{equation}
We use the coordinate matrices of \cref{eq:coordinate-matrices} and write
\begin{equation}
\label{eq:Psi-z}
 \Psi_{\bm z}=\sum_{r=1}^3z_r\Psi_r,\qquad \bm z\in\R^3.
\end{equation}

Let \(0<|\bh|<2\), and put
\[
  \bv=\frac{\bh}{|\bh|},
  \qquad
  r_{\bh}=\sqrt{1-\frac{|\bh|^2}{4}},
\]
and, for \(\bomega\in E_{\bv}\), define the paired aperture points
\begin{equation}
\label{eq:small-circle-pair}
  \bomega_\pm^{\bh}(\bomega)
  =r_{\bh}\bomega\pm\frac{\bh}{2}.
\end{equation}
They lie on the two small circles
\begin{equation}
\label{eq:small-circles}
 C_\pm(\bh)
 =\left\{\boldsymbol\xi\in\Sph^2:
          \boldsymbol\xi\cdot\bv=\pm\frac{|\bh|}{2}\right\},
 \qquad
 \bomega_+^{\bh}(\bomega)-\bomega_-^{\bh}(\bomega)=\bh .
\end{equation}
Conversely, if \(\boldsymbol\xi_\pm\in\Sph^2\) and
\(\boldsymbol\xi_+-\boldsymbol\xi_-=\bh\), then their midpoint is
orthogonal to \(\bh\) and has length \(r_{\bh}\).
Thus \eqref{eq:small-circle-pair} parametrizes all paired points
with chord \(\bh\). Reversing the chord interchanges the two circles.

For a smooth aperture function, define its paired restrictions on
the common parameter circle \(E_{\bv}\) by
\[
 (\mathcal R_\pm^{\bh}f)(\bomega)
 =f\bigl(\bomega_\pm^{\bh}(\bomega)\bigr),
\]
and set
\begin{equation}
\label{eq:finite-small-circle-multiplier}
 \Lambda_{\bh}
 =\diag\!\left(
 e^{i\kappa\bh\cdot\bx_1},\ldots,
 e^{i\kappa\bh\cdot\bx_s}\right),
 \qquad
 \Phi_{\bh}=R^{-1}\Lambda_{\bh}R .
\end{equation}

\begin{proposition}
\label{prop:paired-finite-identity}
For \(Q=\Phi_XR\), the paired restrictions obey
\begin{equation}
\label{eq:finite-small-circle-esprit}
  \mathcal R_+^{\bh}Q
  =\mathcal R_-^{\bh}Q\,\Phi_{\bh}.
\end{equation}
For a fixed unit direction \(\bv\) and \(\bh=\eps\bv\),
\begin{equation}
\label{eq:finite-to-additive-generator}
  \Phi_{\eps\bv}
  =\exp(i\kappa\eps\Psi_{\bv}),
  \qquad
  \Psi_{\bv}=\sum_{r=1}^3v_r\Psi_r.
\end{equation}
\end{proposition}

\begin{proof}
For each atom,
\[
 \varphi_{\bx_j}\bigl(\bomega_+^{\bh}(\bomega)\bigr)
 =
 e^{i\kappa\bh\cdot\bx_j}
 \varphi_{\bx_j}\bigl(\bomega_-^{\bh}(\bomega)\bigr),
\]
which proves \eqref{eq:finite-small-circle-esprit} columnwise and then after
the basis change \(R\).  Since
\(
\Lambda_{\eps\bv}
=\exp(i\kappa\eps D_{\bv}(X))
\)
with \(D_{\bv}(X)=\sum_rv_rD_r(X)\), similarity commutes with the matrix
exponential and gives \eqref{eq:finite-to-additive-generator}.
\end{proof}

Equip \(E_{\bv}\) with normalized arclength \(\dd\sigma/(2\pi)\).
If \(\mathcal R_-^{\bh}Q\) has full column rank in the resulting
space \(L^2(E_{\bv})\), the finite ESPRIT step is
\begin{equation}
\label{eq:small-circle-pseudoinverse}
 \Phi_{\bh}
 =(\mathcal R_-^{\bh}Q)^\dagger\mathcal R_+^{\bh}Q.
\end{equation}
Full rank of the spherical synthesis does not imply full rank of a
circle restriction. Distinct nodes with the same projection onto
\(\bv^\perp\) give proportional restricted columns, since
\[
 (\mathcal R_-^{\eps\bv}\varphi_{\bx_j})(\bomega)
 =
 e^{-i\kappa\eps\bv\cdot\bx_j/2}
 e^{i\kappa r_{\eps\bv}\bomega\cdot
       (I_3-\bv\bv^T)\bx_j}.
\]
When
\begin{equation}
\label{eq:small-circle-no-alias}
 \kappa\eps\max_j|\bv\cdot\bx_j|<\pi,
\end{equation}
the principal matrix logarithm gives the unaliased additive generator
\begin{equation}
\label{eq:small-circle-log}
 \Psi_{\bv}
 =\frac1{i\kappa\eps}\ln\Phi_{\eps\bv}.
\end{equation}
Without \eqref{eq:small-circle-no-alias}, the finite pencil determines the
projected coordinates only modulo \(2\pi/(\kappa\eps)\).

\begin{theorem}
\label{thm:paired-exact-recovery}
Let $Q=\Phi_XR$, and let $\bv_1,\bv_2,\bv_3$ be orthonormal.
Choose $0<\eps_0<\min\{\pi/\kappa,2\}$, put $q_0=\kappa\eps_0$, and
suppose each minus-circle
restriction $\mathcal R_-^{\eps_0\bv_r}Q$ has column rank $s$.
Then, for $r=1,2,3$, its least-squares shift matrix is exactly
\begin{equation}
 (\mathcal R_-^{\eps_0\bv_r}Q)^\dagger
 \mathcal R_+^{\eps_0\bv_r}Q
 =\Phi_{\eps_0\bv_r}.
 \label{eq:paired-exact-shift}
\end{equation}
A generic complex linear combination of the three shift matrices has
simple spectrum and gives a common pairing. If $V$ is a corresponding
common eigenbasis and
$\lambda_{j,r}=(V^{-1}\Phi_{\eps_0\bv_r}V)_{jj}$, then
\begin{equation}
 \bx_j=\frac1{q_0}\sum_{r=1}^3
       \arg(\lambda_{j,r})\bv_r
 \label{eq:paired-exact-readout}
\end{equation}
up to one common permutation. For $s=1$ no pencil search is needed.
\end{theorem}
\begin{proof}
Apply a left inverse to the finite-shift identity to recover each
shift matrix. Since $q_0|\bv_r\cdot\bx_j|<\pi$, the principal phases
are unaliased. Distinct targets have distinct triples of projected
coordinates and hence distinct phase triples. For each pair of
targets, equality of the corresponding pencil eigenvalues defines a
proper hyperplane in the complex coefficient space. Outside the
finite union of these hyperplanes, the pencil has simple spectrum.
Its eigenvectors give a common ordering of the three phase matrices.
The recovery formula follows from orthonormality of the directions.
\end{proof}

The individual shift matrices may have repeated eigenvalues.
Recovery requires distinct joint phase triples, not distinct phases
in every direction.

For $0<\eps<2$, set $q=\kappa\eps$ and
\[
 \nu_\eps=\kappa r_{\eps\bv}
 =\kappa\sqrt{1-\eps^2/4},
 \qquad P_{\bv}=I_3-\bv\bv^T.
\]
Let \(A_{\epsilon\mathbf v}=\mathcal R_-^{\epsilon\mathbf v}\Phi_X:\mathbb C^s\to L^2(E_{\mathbf v})\) be the pullback of the raw minus-circle synthesis, with \(E_{\mathbf v}\) equipped with normalized arclength, i.e.
$$ (A_{\epsilon\mathbf v}c)(\boldsymbol\omega)
=\sum_{j=1}^s c_j
 \exp\!\left(i\kappa\bigl(r_{\epsilon\mathbf v}\boldsymbol\omega
              -\tfrac{\epsilon}{2}\mathbf v\bigr)\cdot\mathbf x_j\right),
\qquad \boldsymbol\omega\in E_{\mathbf v}.$$
 With normalized arclength its Gram matrix is
\begin{equation}
 (A_{\eps\bv}^*A_{\eps\bv})_{jk}
 =e^{-i\kappa\eps\bv\cdot(\bx_k-\bx_j)/2}
 J_0\!\left(\nu_\eps|P_{\bv}(\bx_k-\bx_j)|\right),
 \label{eq:paired-restriction-gram}
\end{equation}
where $J_m$ is the cylindrical Bessel function of the first kind.

Define the circle coherence sum
\begin{equation}
\label{eq:paired-circle-coherence}
 \eta_{\rm circ}(\eps,\bv;X)
 =\max_{1\le j\le s}\sum_{k\ne j}
 \left|J_0\!\left(\nu_\eps
 |P_{\bv}(\bx_k-\bx_j)|\right)\right|.
\end{equation}

\begin{corollary}
\label{cor:paired-circle-conditioning}
The singular values of the raw circle synthesis satisfy
\begin{equation}
\label{eq:paired-circle-singular-bounds}
 1-\eta_{\rm circ}
 \le \sigma_{\min}(A_{\eps\bv})^2
 \le \sigma_{\max}(A_{\eps\bv})^2
 \le 1+\eta_{\rm circ}.
\end{equation}
Consequently, if $\eta_{\rm circ}<1$, then $A_{\eps\bv}$ has rank $s$ and
\begin{equation}
\label{eq:paired-circle-condition-number}
 \cond(A_{\eps\bv})
 \le\sqrt{\frac{1+\eta_{\rm circ}}{1-\eta_{\rm circ}}}.
\end{equation}
For a restricted signal basis $A_{\eps\bv}R$,
$\cond(A_{\eps\bv}R)\le\cond(A_{\eps\bv})\cond(R)$.
\end{corollary}
\begin{proof}
The diagonal entries in \eqref{eq:paired-restriction-gram} equal one,
and the moduli of the off-diagonal entries have the summands in
\eqref{eq:paired-circle-coherence}. The Hermitian Gershgorin theorem
therefore gives \eqref{eq:paired-circle-singular-bounds}; the remaining
claims follow from the relation between Gram eigenvalues and singular
values.
\end{proof}
The bound is sufficient, not necessary: it discards cancellation
between off-diagonal Gram entries. It also shows why three-dimensional
minimum separation alone does not determine circle conditioning.
The relevant geometry is that of the projected cloud.

\begin{proposition}
\label{prop:paired-rank}
For $0<\eps<2$ (equivalently, $0<q=\kappa\eps<2\kappa$), the continuous
restriction $A_{\eps\bv}$ has column
rank $s$ if and only if the projected points $P_{\bv}\bx_j$ are distinct.
When they are distinct, a generic collection of at least $s$ circle
sample angles also gives column rank $s$.
\end{proposition}
\begin{proof}
Necessity follows from proportional columns. For sufficiency, a linear
dependence would give a zero far field for a two-dimensional radiating
Helmholtz field with point sources at the distinct projected points,
after absorbing the nonzero axial phases into their coefficients.
Rellich uniqueness and continuation imply the field vanishes away
from the sources, and its distinct logarithmic singularities force
all coefficients to vanish; see \cite{coltonkress2019} for the underlying
uniqueness principles. The trace functions are analytic. Their linear
independence implies some $s$ evaluation angles give a nonzero
determinant; that analytic minor vanishes only on a measure-zero set
with empty interior. The same conclusion holds with more sample angles.
\end{proof}
For every finite distinct cloud, a generic orthonormal triple of
circle axes gives full column rank for all three continuous restrictions.

\begin{remark}[Paired arcs in a hemisphere]
\label{rem:hemispherical-paired-arcs}
Let $\Gamma=\{\bomega\in\Sph^2:\omega_3>0\}$. Since
$C_-(\bh)=-C_+(\bh)$, a complete pair of small circles cannot lie in
$\Gamma$. For $\bh=\eps\bv$, however, both paired points
$r_{\bh}\boldsymbol\eta\pm\eps\bv/2$, $\boldsymbol\eta\in E_{\bv}$,
belong to $\Gamma$ on the common parameter arc
\[
 r_{\bh}\eta_3>\eps|v_3|/2.
\]
Because $\max_{\boldsymbol\eta\in E_{\bv}}\eta_3=\sqrt{1-v_3^2}$,
this arc is nonempty and open precisely when
\begin{equation}
 0<\eps<2\sqrt{1-v_3^2}.
 \label{eq:hemispherical-arc-availability}
\end{equation}
The finite-shift identity \eqref{eq:finite-small-circle-esprit} holds
unchanged on the arc. If the projected nodes are distinct, an analytic
linear combination of the restricted atoms that vanishes on the arc
vanishes on the full circle; \cref{prop:paired-rank} therefore gives
continuous rank $s$, and at least $s$ generic paired samples give exact
rank. With one common signal basis on $\Gamma$, three generic tilted
orthonormal axes admit recovery for sufficiently small unaliased shifts;
larger shifts require consistent phase continuation. North-axis shifts
are unavailable, and short arcs can be poorly conditioned. Along a
nonhorizontal tilted axis, increasing the displacement shortens the
observable arc, trading restriction conditioning against phase sensitivity.
\end{remark}

\section{Whole-sphere generator ESPRIT}
\label{sec:whole-sphere-method}

The generator equations determine the three coordinate matrices
through one coupled system. We prove that the continuous system is
injective whenever the nodes are distinct.
Section~\ref{sec:finite-realization} treats finite harmonic
realization and its additional rank requirements.

The coordinate multipliers are $M_rf(\bomega)=\omega_rf(\bomega)$.
We denote the tangent-plane projection by
\begin{equation}
\label{eq:tangent-projection}
 \bPi_{\bomega}=I_3-\bomega\bomega^T.
\end{equation}
Define the angular-momentum generators by
\[
  \mathcal L_r=-i(\be_r\times\bomega)\cdot
       \nabla_{\Sph^2}=-i(\cJ \nabla_{\Sph^2})_r,\qquad r=1,2,3
\]
where 
\begin{equation}
\label{eq:Hodge-map}
  (\cJ F)(\bomega)=\bomega\times F(\bomega),\qquad  F\in L^2(\Sph^2;\C^{3}).
\end{equation}

On each tangent plane, \(\cJ\) is rotation by \(\pi/2\), so
\begin{equation}
\label{eq:Hodge-properties}
  \cJ ^*=-\cJ ,
  \qquad
  \cJ ^*\cJ =I,
  \qquad
  \cJ ^{-1}=-\cJ .
\end{equation}

For \(G\in\SO(3)\), define the unitary left action
\begin{equation}
\label{eq:rotation-representation}
  (U(G)f)(\bomega)=f(G^{-1}\bomega).
\end{equation}
For \(\ba\in\R^3\), let \([\ba]_\times\bu=\ba\times\bu\), set
\[
  R_{\ba}(\theta)=\exp(\theta[\ba]_\times),
  \qquad
  \mathcal L_{\ba}=\sum_{i=1}^3a_i\mathcal L_i
  =-i(\ba\times\bomega)\cdot\nabla_{\Sph^2}.
\]
For smooth $f$, differentiation at $\theta=0$ gives
$\partial_\theta U(R_{\ba}(\theta))f|_0=-i\mathcal L_{\ba}f$, hence
\begin{equation}
\label{eq:one-parameter-unitary}
 U(R_{\ba}(\theta))=e^{-i\theta\mathcal L_{\ba}}.
\end{equation}
The forward point orbit has the opposite sign:
\begin{equation}
\label{eq:forward-orbit-sign}
 Q(R_{\ba}(\theta)\bomega)
 =(e^{i\theta\mathcal L_{\ba}}Q)(\bomega),
 \qquad
 \left.\partial_\theta Q(R_{\ba}(\theta)\bomega)\right|_0
 =i\mathcal L_{\ba}Q(\bomega).
\end{equation}
For $\ba\ne0$, the field $K_{\ba}(\bomega)=\ba\times\bomega$ is
 tangent to the rotation orbits and vanishes at $\pm\ba/|\ba|$.

The aperture commutators are given in
Appendix~\ref{app:aperture-algebra}.

\begin{proposition}
\label{prop:differential-lie-cocycle-bridge}
Let \(Q=\Phi_XR\) and
\(\Psi_r=R^{-1}D_r(X)R\).  Then:
\begin{enumerate}[label=(\roman*),leftmargin=*]
\item For every
\(\boldsymbol\tau\in T_{\bomega}\Sph^2\),
\begin{equation}
\label{eq:local-tangent-pencil}
  \boldsymbol\tau\cdot\nabla_{\Sph^2}Q(\bomega)
  =i\kappa Q(\bomega)\Psi_{\boldsymbol\tau}.
\end{equation}
For the angular-velocity vector
\(
  \ba_{\bomega,\boldsymbol\tau}
  :=\bomega\times\boldsymbol\tau
\)
at the base point \(\bomega\), we have
\(
  \ba_{\bomega,\boldsymbol\tau}\times\bomega
  =\boldsymbol\tau
\), and hence
\begin{equation}
%\label{eq:local-pencil}
\label{eq:tangent-as-angular-generator}
  \boldsymbol\tau\cdot\nabla_{\Sph^2}Q(\bomega)
  =i \mathcal L_{\bm\omega\times \bm\tau}Q(\bomega).
\end{equation}

\item If instead \(\ba\) is a fixed rotation axis over the whole sphere,
then its tangent velocity is
\(
  K_{\ba}(\bomega)=\ba\times\bomega
\), and
\begin{equation}
\label{eq:fixed-axis-variable-pencil}
 \mathcal L_{\ba}Q(\bomega)
  =\kappa Q(\bomega)\Psi_{\ba\times\bomega}.
\end{equation}
In particular, with the Killing fields
\(
  K_i(\bomega):=\be_i\times\bomega
\),
\begin{equation}
\label{eq:Q-intertwining}
%\label{eq:fixed-axis-component-system}
  \frac1\kappa\mathcal L_iQ
  =Q\Psi_{K_i(\bomega)}
  =\sum_{q,r=1}^3\epsilon_{iqr}M_qQ\Psi_r.
\end{equation}
The Killing fields span the tangent plane with one redundancy. At each
base point, setting
$\ba_{\bomega,\boldsymbol\tau}=\bomega\times\boldsymbol\tau$ gives
\begin{equation}
\label{eq:moving-frame-reconstruction}
\begin{aligned}
 &\sum_i\omega_iK_i(\bomega)=0,
 \qquad
 \boldsymbol\tau=\sum_i
   (\ba_{\bomega,\boldsymbol\tau})_iK_i(\bomega),\\
 &\boldsymbol\tau\cdot\nabla_{\Sph^2}Q(\bomega)
 =i\sum_i(\ba_{\bomega,\boldsymbol\tau})_i
               \mathcal L_iQ(\bomega).
\end{aligned}
\end{equation}
\end{enumerate}
\end{proposition}
\begin{proof}
For a spherical atom and a tangent vector \(\boldsymbol\tau\),
\[
  \boldsymbol\tau\cdot\nabla_{\Sph^2}
  e^{i\kappa\bomega\cdot\bx}
  =i\kappa(\boldsymbol\tau\cdot\bx)
   e^{i\kappa\bomega\cdot\bx}.
\]
Applying this identity columnwise and changing basis by \(R\) proves
\eqref{eq:local-tangent-pencil}.  The vector triple-product identity gives
\[
  (\bomega\times\boldsymbol\tau)\times\bomega
  =\boldsymbol\tau,
\]
which proves \eqref{eq:tangent-as-angular-generator}.  Taking
\(\boldsymbol\tau=\ba\times\bomega\) proves
\eqref{eq:fixed-axis-variable-pencil}; expanding the cross product gives
\eqref{eq:Q-intertwining}.  The identities in \eqref{eq:moving-frame-reconstruction} follow from
\(
  \sum_i\omega_i\be_i\times\bomega=0
\)
and
\(
  (\bomega\times\boldsymbol\tau)\times\bomega
  =\boldsymbol\tau
\).
\end{proof}

Define the stacked angular operator and coupled design by
\begin{equation}
\label{eq:BQ-block}
 \cL=\begin{bmatrix}\mathcal L_1\\\mathcal L_2\\\mathcal L_3\end{bmatrix},
 \qquad
 \cB_Q=\begin{bmatrix}
 0&-M_3Q&M_2Q\\
 M_3Q&0&-M_1Q\\
 -M_2Q&M_1Q&0
 \end{bmatrix}.
\end{equation}
For a stacked trial triple $Z=(Z_1^T,Z_2^T,Z_3^T)^T$, its geometric
form is
\begin{equation}
\label{eq:qZ}
 \cQ_Z(\bomega)=\sum_{r=1}^3\be_rQ(\bomega)Z_r,
 \qquad \cB_QZ=\cJ\cQ_Z.
\end{equation}
Thus \eqref{eq:Q-intertwining} becomes
\begin{equation}
\label{eq:main-display}
 \cL Q=\kappa\cB_Q\Psi,
 \qquad
 \Psi=\begin{bmatrix}\Psi_1\\\Psi_2\\\Psi_3\end{bmatrix}.
\end{equation}
We next prove that the coupled design is injective. Once $\Psi$ is
known, a generic real pencil gives a common eigenbasis and thereby
pairs the three coordinates.

\subsection{Angular block injectivity and Gram equations}
\label{sec:identifiability}

To formulate injectivity independently of the signal basis, define,
for \(\bb=(\bb_1,\ldots,\bb_s)\in(\C^3)^s\),
\begin{equation}
\label{eq:node-tangential-design}
  (\cT_X\bb)(\bomega)
  =\bPi_{\bomega}\sum_{j=1}^s
    \bb_j e^{i\kappa\bomega\cdot\bx_j}.
\end{equation}

\begin{theorem}
\label{thm:qualitative-injectivity}
If \(\kappa>0\) and \(\bx_1,\ldots,\bx_s\) are distinct, then
\(\cT_X\) is injective.  The angular node design obtained by
replacing \(\bPi_{\bomega}\) with \(\bomega\times\) is injective as well.
Consequently, for every \(R\in GL(s,\C)\), the matrix block designs
\(\cT_Q\) and \(\cB_Q\), with \(Q=\Phi_XR\), are injective.
\end{theorem}

\begin{proof}
Suppose
\[
  \bomega\times
  \sum_{j=1}^s\bb_j e^{i\kappa\bomega\cdot\bx_j}=0
  \qquad\text{for every }\bomega\in\Sph^2.
\]
Define the vector field
\[
  \bm u(\bm y)
  =\sum_{j=1}^sG_\kappa(\bm y-\bx_j)\bb_j,\qquad G_\kappa(\bm y):=\frac{e^{i\kappa|\bm y|}}{4\pi|\bm y|}
\]
where $G_\kappa$ is  the outgoing Helmholtz fundamental solution. 
Its far-field pattern is a nonzero scalar multiple of
\(
\sum_j\bb_j e^{-i\kappa\widehat{\bm y}\cdot\bx_j}
\).
The far-field pattern of \(\nabla\times\bm u\) is therefore a nonzero
scalar multiple of
\[
  \widehat{\bm y}\times
  \sum_j\bb_j e^{-i\kappa\widehat{\bm y}\cdot\bx_j},
\]
which vanishes by the hypothesis with
\(\bomega=-\widehat{\bm y}\).  The Rellich lemma for radiating
Helmholtz fields and unique continuation imply
\(
\nabla\times\bm u=0
\)
on the connected set \(\R^3\setminus X\); see, for example,
\cite{coltonkress2019}.

Near \(\bx_j\), put \(\bm y=\bx_j+r\widehat{\bm z}\).  The
contribution from the \(j\)-th source has the singular expansion
\[
  \nabla G_\kappa(r\widehat{\bm z})\times\bb_j
  =-\frac{\widehat{\bm z}\times\bb_j}{4\pi r^2}
   +O(r^{-1}),
\]
while all other source contributions are smooth.  Since the curl vanishes,
multiplication by \(r^2\) followed by \(r\downarrow0\) gives
\(
\widehat{\bm z}\times\bb_j=0
\)
for every \(\widehat{\bm z}\in\Sph^2\), hence \(\bb_j=0\).  This holds
for all \(j\), proving node-design injectivity.  The pointwise identity $|\bomega\times\bw|=|P_{\bomega}\bw|$ proves the
tangential statement.  Applying the node result separately to each output
column of \(QZ=\Phi_XRZ\), and using invertibility of \(R\), proves
matrix-block injectivity.
\end{proof}

For a stacked matrix-valued field $Y$, write
\begin{equation}
\label{eq:stacked-L2-Frobenius-norm}
 \|Y\|_{L_{\rm F}^2}^2
 =\int_{\Sph^2}\|Y(\bomega)\|_{\rm F}^2
   \frac{\dd\bomega}{4\pi}.
\end{equation}

\begin{theorem}
\label{thm:exact-whole-sphere}
Let \(Q=\Phi_XR\), where \(\Phi_X\) has full column rank and
\(R\in GL(s,\C)\).  Then under the assumptions of \cref{thm:qualitative-injectivity}
\begin{equation}
\label{eq:whole-sphere-LS}
  \widehat Z
  =\argmin_{Z\in(\C^{s\times s})^3}
  \| \cL Q -\kappa\cB_QZ\|_{L_{\rm F}^2}^2
\end{equation}
has the unique solution
\begin{equation}
\label{eq:whole-sphere-solution}
 \widehat Z=\kappa^{-1}\cB_Q^\dagger\cL Q=\Psi,
 \qquad
 \widehat Z_r=R^{-1}D_r(X)R,\quad r=1,2,3.
\end{equation}
The recovered matrices therefore satisfy
\begin{equation}
\label{eq:commuting-joint-spectrum}
  [\Psi_r,\Psi_t]=0,
  \qquad
  \operatorname{jspec}(\Psi_1,\Psi_2,\Psi_3)=X.
\end{equation}
If the signal basis is changed to \(QS\), \(S\in GL(s,\C)\), then
\begin{equation}
\label{eq:basis-equivariance}
  \Psi_r(QS)=S^{-1}\Psi_r(Q)S.
\end{equation}
Thus the recovered joint spectrum depends only on the signal subspace.
\end{theorem}

\begin{proof}
By \eqref{eq:main-display}, the coordinate triple has zero residual.
Block injectivity makes it the unique minimizer and gives the
pseudoinverse formula.  Simultaneous similarity to the
three diagonal matrices proves commutativity and identifies the joint
spectrum.  Replacing \(R\) by \(RS\) gives
\(
(RS)^{-1}D_r(X)(RS)=S^{-1}\Psi_rS
\), proving equivariance.
\end{proof}

The Gram operator gives an equivalent characterization of block
injectivity.
\begin{proposition}
\label{prop:Gram-blocks}
The common Gram operator
\(
\Gamma_Q=\cT_Q^*\cT_Q=\cB_Q^*\cB_Q
\)
acts by
\begin{equation}
\label{eq:Gram-action}
  (\Gamma_QZ)_r=\sum_{t=1}^3H_{rt}(Q)Z_t
\end{equation}
with
\begin{equation}
\label{eq:Gram-blocks-Q}
  H_{rt}(Q)=
  \int_{\Sph^2}
  (\delta_{rt}-\omega_r\omega_t)
  Q(\bomega)^*Q(\bomega)
  \frac{\dd\bomega}{4\pi},\qquad r,t\in\{1,2,3\}. 
\end{equation}
In particular,
\begin{equation}
\label{eq:Gram-quadratic}
  \langle Z,\Gamma_QZ\rangle
  =\int_{\Sph^2}
    \|\bPi_{\bomega}\cQ_Z (\bomega)\|_{\rm F}^2
    \frac{\dd\bomega}{4\pi}
  =\int_{\Sph^2}
    \|\cJ\cQ_Z (\bomega)\|_{\rm F}^2
    \frac{\dd\bomega}{4\pi}.
\end{equation}
Thus \(\cT_Q\) and \(\cB_Q\) are injective if and only if
\(\lambda_{\min}(\Gamma_Q)>0\).
\end{proposition}

\begin{proof}
Expand the squared norm in \eqref{eq:Gram-quadratic} and use
\(
\bPi_{\bomega}^*\bPi_{\bomega}=\bPi_{\bomega}
\).
The angular equality follows from
\(
\cJ^*\cJ Z(\bomega)=\bPi_{\bomega}Z(\bomega)
\).
Reading off the coefficient of each \(Z_r\) gives
\eqref{eq:Gram-action}.
\end{proof}

The corresponding node Gram matrix can be expressed through the
matrix kernel
\begin{equation}
\label{eq:tangential-kernel}
 T_\kappa(\bd)
 =\int_{\Sph^2}\bPi_{\bomega}
 e^{i\kappa\bomega\cdot\bd}\frac{\dd\bomega}{4\pi}.
\end{equation}
The block Gram matrix $G_X=\cT_X^*\cT_X$ on $(\C^3)^s$ has blocks
\begin{equation}
\label{eq:whole-sphere-node-gram-blocks}
 (G_X)_{jk}=T_\kappa(\bx_k-\bx_j),
 \qquad (G_X)_{jj}=T_\kappa(0)=\frac23I_3.
\end{equation}
Set
\begin{equation}
\label{eq:whole-sphere-block-coherence}
 \eta_{\rm sph}(\kappa;X)
 =\max_{1\le j\le s}\sum_{k\ne j}
 \|T_\kappa(\bx_k-\bx_j)\|_2.
\end{equation}

\begin{corollary}
\label{cor:whole-sphere-conditioning}
The tangential node design and its angular Hodge rotation have the same
singular values, and
\begin{equation}
\label{eq:whole-sphere-singular-bounds}
 \frac23-\eta_{\rm sph}
 \le\sigma_{\min}(\cT_X)^2
 \le\sigma_{\max}(\cT_X)^2
 \le\frac23+\eta_{\rm sph}.
\end{equation}
In particular, if $\eta_{\rm sph}<2/3$, then
\begin{equation}
\label{eq:whole-sphere-condition-number}
 \cond(\cT_X)
 \le\sqrt{\frac{\frac23+\eta_{\rm sph}}
                   {\frac23-\eta_{\rm sph}}}.
\end{equation}
\end{corollary}
\begin{proof}
For $\bb=(\bb_1,\ldots,\bb_s)$, expansion of the node-design norm gives
\[
 \|\cT_X\bb\|_2^2
 =\frac23\sum_j|\bb_j|^2
  +\sum_{j\ne k}\bb_j^*T_\kappa(\bx_k-\bx_j)\bb_k.
\]
Since $T_\kappa(-\bd)=T_\kappa(\bd)^*$, the block Gram matrix is
Hermitian. Bounding the cross terms by their operator norms and using
$2ab\le a^2+b^2$ gives
\[
 \left|\sum_{j\ne k}\bb_j^*T_\kappa(\bx_k-\bx_j)\bb_k\right|
 \le\eta_{\rm sph}\sum_j|\bb_j|^2.
\]
This proves the bounds. The pointwise Hodge isometry gives the same Gram
matrix for the angular design.
\end{proof}
For $\bd\ne0$, with $t=\kappa|\bd|$ and $\bu=\bd/|\bd|$, the kernel is
\begin{equation}
\label{eq:tangential-kernel-formula}
 T_\kappa(\bd)
 =\left(\sinc(t)+\frac{\sinc'(t)}t\right)(I_3-\bu\bu^T)
   -\frac{2\sinc'(t)}t\bu\bu^T.
\end{equation}
Here $\sinc(t)=\sin(t)/t$ for $t\ne0$ and $\sinc(0)=1$.
Hence
\[
 \|T_\kappa(\bd)\|_2
 =\max\left\{
 \left|\sinc(t)+\frac{\sinc'(t)}t\right|,
 \left|\frac{2\sinc'(t)}t\right|\right\}.
\]
As in the circle estimate, the bound discards cancellation among
off-diagonal interactions.

\begin{corollary}
\label{cor:atom-independence}
For distinct nodes and \(\kappa>0\), the synthesis operator \(\Phi_X\)
has full column rank.  Consequently the hypotheses of
\cref{thm:exact-whole-sphere} hold at continuous full aperture.
\end{corollary}
\begin{proof}
If \(\Phi_X\bc=0\), choose a fixed nonzero \(\bu\in\C^3\) and put
\(\bb_j=c_j\bu\).  Then \(\cT_X\bb=0\), so
\cref{thm:qualitative-injectivity} gives \(\bb_j=0\) and hence
\(c_j=0\) for every \(j\).
\end{proof}

\section{Hodge equivalence and the synchronized limit}
\label{sec:methods-relation}

The angular-generator and projected-gradient formulations give the
same whole-sphere least-squares estimate, even for a smooth empirical
frame. We prove this equivalence by a tangential Hodge isometry.
We then derive the local generator identity as a synchronized
differential limit of the paired-circle identity.

%\subsection{Hodge equivalence with tangential differential ESPRIT}
%\label{sec:hodge-equivalence}

Define the tangential design and intrinsic differential data by
\begin{equation}
\label{eq:TQ-GQ}
  \cT_QZ=\bPi_{\bomega}\cQ_Z ,
  \qquad
   G_Q=-i\nabla_{\Sph^2}Q.
\end{equation}

For $Q=\Phi_XR$ we have the equivalent intrinsic-gradient identity
\begin{equation}
\label{eq:gradient-intertwining}
-i\nabla_{\Sph^2}Q(\bomega)
  =\kappa\sum_{r=1}^3\bPi_{\bomega}\be_r\,Q(\bomega)\Psi_r
\end{equation}
with $\bPi_{\bomega}$ as in \cref{eq:tangent-projection}.

Both \(\cT_QZ\) and \( G_Q\) belong to the tangent-field space
\[
  L_T^2=\left\{
    F\in L^2(\Sph^2;\C^{3\times s}):
    \bomega^TF(\bomega)=0\ \text{a.e.}
  \right\}.
\]

\begin{theorem}
\label{thm:Hodge-equivalence}
For every trial triple \(Z\),
\begin{equation}
\label{eq:Hodge-factorizations}
  \cB_QZ=\cJ \cT_QZ,
  \qquad
   \cL Q =\cJ  G_Q,
\end{equation}
and therefore
\begin{equation}
\label{eq:residual-isometry}
   \cL Q -\kappa\cB_QZ
  =\cJ (G_Q-\kappa\cT_QZ).
\end{equation}
Consequently,
\begin{align}
  \cB_Q^*\cB_Q=\cT_Q^*\cT_Q,
  \qquad\cB_Q^* \cL Q =\cT_Q^* G_Q,\qquad
  \cB_Q^\dagger=\cT_Q^\dagger\cJ ^*,\\
  \widehat\Psi^{\rm ang}
  :=\kappa^{-1}\cB_Q^\dagger \cL Q
  =\kappa^{-1}\cT_Q^\dagger G_Q
  =:\widehat\Psi^{\rm tan}.
  \label{eq:equal-estimators}
\end{align}
\end{theorem}

\begin{proof}
For every \(\bm z\in\C^3\),
\[
  \bomega\times\bm z
  =\bomega\times\bPi_{\bomega}\bm z.
\]
This gives the first factorization. For a smooth frame $Q$,
the identity $\boldsymbol{\mathcal L}=-i\cJ\nabla_{\Sph^2}$ gives
\[
 \cL Q =-i\cJ\nabla_{\Sph^2}Q=\cJ G_Q.
\]
No finite-atom assumption is needed for this factorization.
Subtracting the two identities gives the residual relation.
Since \(\cJ\) is unitary on \(L_T^2\), the normal equations and
pseudoinverse identities follow.
\end{proof}
For an exact frame $Q=\Phi_XR$, both estimates in
\eqref{eq:equal-estimators} equal $\Psi$. Without the factor
$\kappa^{-1}$, the pseudoinverse output is $\kappa\Psi$.
Thus continuous whole-sphere differential ESPRIT and Lie-algebraic
ESPRIT are the same least-squares problem, expressed in two tangent
frames:
\[
  \text{projected gradient }\bPi_{\bomega}\cQ_Z 
  \quad\longleftrightarrow\quad
  \text{angular field }\cJ\cQ_Z .
\]
The angular formulation is convenient for harmonic discretization:
each \(\mathcal L_i\) preserves harmonic degree and has a sparse
matrix representation.

At a single aperture point, both designs have rank two:
\begin{equation}
\label{eq:pointwise-rank-two}
  \ker\bPi_{\bomega}
  =\ker(\bomega\times)=\operatorname{span}\{\bomega\},
  \qquad
  \sum_{i=1}^3\omega_i\mathcal L_iQ(\bomega)=0.
\end{equation}
The missing normal line varies with \(\bomega\).  In particular,
\begin{equation}
\label{eq:average-projector}
  \int_{\Sph^2}\bPi_{\bomega}\frac{\dd\bomega}{4\pi}
  =\frac23I_3.
\end{equation}
For a multi-atom design, full-aperture injectivity additionally follows from
\cref{thm:qualitative-injectivity}.

\begin{remark}\label{rem:discrete-equivalence}
For sampled directions \(\bomega_n\) with positive weights \(w_n\), let
\[
 W_h=\diag(w_1,\ldots,w_N)\otimes I_3,
 \qquad
 \cJ _h=\diag\bigl(
 [\bomega_1]_\times,\ldots,[\bomega_N]_\times\bigr).
\]
When the sum models the normalized surface measure
\(\dd\bomega/(4\pi)\), take \(\sum_nw_n=1\).
For a matrix-valued vector field \(F\), write
\(\operatorname{stack}_nF(\bomega_n)\) for node-major vertical stacking,
and define
\[
  \cT_hZ=\operatorname{stack}_{n=1}^N
      (\cT_QZ)(\bomega_n)\in\C^{3N\times s},
  \qquad
  \bm G_h=\operatorname{stack}_{n=1}^N
      \frac1{i}\nabla_{\Sph^2}Q(\bomega_n)
      \in\C^{3N\times s}.
\]
The sampled angular block and data are then
\[
  \cB_hZ=\cJ _h\cT_hZ,
  \qquad
   Y_h=\cJ _h\bm G_h.
\]
Exact discrete equivalence persists when
\[
  \cJ _h^*W_h\cJ _h=W_h
\]
on the discrete tangent space.  It may fail if the two derivatives are
discretized independently, radial leakage is not removed, or different
component weights or regularizers are used.  If an output projector
\(P_h\) preserves the discrete tangent space, define both angular
quantities by applying \(\cJ _h\) after \(P_h\) to retain the
isometry.  Applying \(P_h\) after the Hodge map instead requires
\(P_h\cJ _h=\cJ _hP_h\) for the two constructions to agree.
\end{remark}

%\subsection{Paired small circles and their great-circle limit}
%\label{sec:paired-circles}
%\subsection{Based tangent directions and the incidence manifold}
%\label{sec:incidence}

A tangent direction must be specified together with its base point
on the sphere. The unit-tangent incidence manifold is
\begin{equation}
\label{eq:incidence-manifold}
  \mathfrak F
  =\{(\bomega,\bv)\in\Sph^2\times\Sph^2:
       \bomega\cdot\bv=0\}
  \simeq\SO(3),
\end{equation}
where \((\bomega,\bv)\) is identified with the oriented frame
\([\bomega\quad \bv\quad \bomega\times\bv]\). 

Fixing \(\bv\) in \(\mathfrak F\) gives the equator $E_{\bv}$.
Along this equator, the derivative in direction \(\bv\) involves
the same coefficient matrix \(\Psi_{\bv}\).
Fixing \(\bomega\) instead gives all unit tangent directions at that
point. The paired-circle differential construction uses the first
description; the whole-sphere construction collects the second over
all base points.

Write the translation-character matrix as
\begin{equation}
\label{eq:Phi-h}
 \Phi(\bh)=R^{-1}\diag(e^{i\kappa\bh\cdot\bx_j})R
          =e^{i\kappa\Psi_{\bh}}.
\end{equation}
Its group law and infinitesimal generator are
\begin{equation}
\label{eq:Phi-group-law}
 \Phi(\bh_1+\bh_2)=\Phi(\bh_1)\Phi(\bh_2),
 \qquad
 \left.\frac1{i\kappa}\partial_t\Phi(t\bh)\right|_{t=0}
 =\Psi_{\bh},
\end{equation}
and the plane-wave identity gives
\begin{equation}
\label{eq:any-two-points}
 Q(\bomega')=Q(\bomega)\Phi(\bomega'-\bomega).
\end{equation}
For $0<t<1$ and $\boldsymbol\eta\in E_{\bv}$, use the paired paths
\begin{equation}
\label{eq:paired-points}
 \bomega_\pm(t,\boldsymbol\eta)
 =\sqrt{1-t^2}\,\boldsymbol\eta\pm t\bv.
\end{equation}
Their chord is $2t\bv$, so \eqref{eq:finite-small-circle-esprit} gives
\begin{equation}
\label{eq:paired-circle-pencil}
 Q(\bomega_+)=Q(\bomega_-)e^{2i\kappa t\Psi_{\bv}}.
\end{equation}
With $\ba_{\boldsymbol\eta}=\boldsymbol\eta\times\bv$ and
$\alpha=\arcsin t$, Rodrigues' formula yields
\begin{equation}
\label{eq:paired-rotations}
 \bomega_\pm=R_{\ba_{\boldsymbol\eta}}(\pm\alpha)\boldsymbol\eta,
 \qquad R_{\ba_{\boldsymbol\eta}}(2\alpha)\bomega_-=\bomega_+.
\end{equation}
Thus the paired identity synchronizes rotations whose axes vary along
the equator. In the cocycle notation of \cref{sec:compatibility},
\begin{equation}
\label{eq:synchronized-cocycle}
 \cC_{\ba_{\boldsymbol\eta}}(2\alpha,\bomega_-)
 =\Phi(2t\bv)=e^{2i\kappa t\Psi_{\bv}}.
\end{equation}
The constant chord makes the coefficient matrix independent of
$\boldsymbol\eta$.

\begin{theorem}
\label{thm:great-circle-limit}
For every \(\boldsymbol\eta\in E_{\bv}\),
\begin{equation}
\label{eq:great-circle-limit}
  \lim_{t\to0}
  \frac{Q(\bomega_+(t,\boldsymbol\eta))
       -Q(\bomega_-(t,\boldsymbol\eta))}
       {2it}
  =\kappa Q(\boldsymbol\eta)\Psi_{\bv} =
 \mathcal L_{\boldsymbol\eta\times\bv}Q
      (\boldsymbol\eta).
\end{equation}
Moreover,
\begin{equation}
\label{eq:pencil-generator-limit}
  \lim_{t\to0}
  \frac{e^{2i\kappa t\Psi_{\bv}}-I_s}{2i\kappa t}
  =\Psi_{\bv}.
\end{equation}
\end{theorem}

\begin{proof}
The curves in \eqref{eq:paired-points} have velocities
\(\pm\bv\) at \(t=0\).  Since \(\bv\in T_{\boldsymbol\eta}\Sph^2\), the
symmetric difference converges to
\(i^{-1}\bv\cdot\nabla_{\Sph^2}Q\).  Apply
\cref{prop:differential-lie-cocycle-bridge}.  The matrix limit follows by differentiating the
exponential.
\end{proof}

Collecting \eqref{eq:great-circle-limit} over all \(\bv\) and all
base points on \(E_{\bv}\) gives the local tangent identities on
\(\mathfrak F\). Expressing the tangent directions through the
fields \(K_i\) then yields the whole-sphere system
\eqref{eq:main-display}. In this sense, the paired-circle
differential identities are equatorial restrictions of the
whole-sphere generator identity.

\section{Guarded harmonic realization}
%\label{sec:harmonics}
%\section{Finite realization and identifiability}
\label{sec:finite-realization}

The exact atoms are smooth, but an empirical signal subspace requires
additional regularity to support differentiation or restriction.
Angular differentiation is unbounded on $L^2(\Sph^2)$, and restriction
to a circle is not defined for a general $L^2$ equivalence class.
A small global subspace error therefore does not control either
operation. A small high-frequency component can have a large
derivative, while a perturbation concentrated near a circle can have
a large trace. Neither finite signal rank nor smoothness without
a quantitative bound removes these obstructions.

A fixed harmonic cutoff makes both operations bounded. The resulting
constants depend on the cutoff: differentiation introduces the angular
bandwidth relative to $\kappa$, and restriction depends on the possible
concentration near the selected circle. Increasing the cutoff can
reduce representation error while admitting more sensitive noise
components. Alternatively, one can impose Sobolev bounds:
$H^1$ controls angular differentiation, and $H^\sigma$ with
$\sigma>1/2$ controls restriction to a fixed smooth circle.

For the generator method, projecting onto a finite harmonic space
before differentiation makes the derivative bounded. Retaining one
additional harmonic degree also preserves the projected generator
identity for exact atoms. Bandlimiting makes a circle restriction
bounded as well, but the restricted finite expansion need not satisfy
the original plane-wave shift identity exactly. Directly acquired
circle data can instead be assessed in the circle norm. In either
case, the linear system and the common eigenvalue calculation must be
well conditioned. The circle method also needs a correct phase branch.

Let
\begin{equation}
\label{eq:harmonic-space}
  \cH_L=\operatorname{span}\{Y_\ell^m:
       0\le\ell\le L,\ |m|\le\ell\},
  \qquad \dim\cH_L=(L+1)^2,
\end{equation}
where the spherical harmonics are orthonormal for the normalized measure
\(\dd\bomega/(4\pi)\), and let \(\Pi_L\) be the corresponding orthogonal
projection.  Define
\begin{equation}
\label{eq:truncated-synthesis}
  \Phi_{X,L}=\Pi_L\Phi_X,
  \qquad
  G_{X,L}=\Phi_{X,L}^*\Phi_{X,L}
           =\Phi_X^*\Pi_L\Phi_X.
\end{equation}
Assume that \(\Phi_{X,L}\) has full column rank.  Choose
\(R_L\in GL(s,\C)\) so that
\nopagebreak[4]
\begin{equation}
\label{eq:finite-orthonormal-frame}
  Q_L=\Phi_{X,L}R_L,
  \qquad Q_L^*Q_L=I_s,
  \qquad R_L^*G_{X,L}R_L=I_s.
\end{equation}
Define the finite coordinate matrices by
\begin{equation}
\label{eq:finite-coordinate-matrices}
  \Psi_{r,L}=R_L^{-1}D_r(X)R_L,
  \qquad r=1,2,3.
\end{equation}
We use the stacked and directional notation
\begin{equation}
\label{eq:finite-coordinate-stack}
 \Psi_L=
 \begin{bmatrix}\Psi_{1,L}\\\Psi_{2,L}\\\Psi_{3,L}\end{bmatrix},
 \qquad
 \Psi_{\bgamma,L}=\sum_{r=1}^3\gamma_r\Psi_{r,L},
 \qquad \bgamma\in\R^3.
\end{equation}
The triple still has joint spectrum \(X\).  

Choose an integer retained degree \(1\le K\le L-1\).  The degrees
\(K+1,\ldots,L\) are guard shells.  One guard shell is algebraically
sufficient; extra shells can protect against approximate transforms or
additional filtering.  Let \(C_L\in\C^{(L+1)^2\times s}\) be the harmonic
coefficient matrix of \(Q_L\), so \(C_L^*C_L=I_s\).  Let \(S_{K,L}\)
be the row selector that restricts harmonic coefficients from degrees
\(0,\ldots,L\) to degrees \(0,\ldots,K\), and set
\(C_K=S_{K,L}C_L\).  Define
\begin{equation}
\label{eq:finite-operators}
  L_i^{(K)}=\mathcal L_i|_{\cH_K},
  \qquad
  M_q^{K,L}=\Pi_KM_q|_{\cH_L}.
\end{equation}

\begin{theorem}
\label{thm:guard-shell} 
The finite coefficient matrices exactly satisfy
\begin{equation}
\label{eq:guard-shell-identity}
 Y_K=\mathbb B_{K,L}\Psi_L,
\qquad 1\le K\le L-1,
\end{equation}
where
\begin{equation}
\label{eq:finite-data-design}
\begin{split}
  Y_K&=\frac1\kappa
 \begin{bmatrix}L_1^{(K)}C_K\\L_2^{(K)}C_K\\L_3^{(K)}C_K\end{bmatrix},\\[1mm]
 \mathbb B_{K,L}&=
 \begin{bmatrix}
 0&-M_3^{K,L}C_L&M_2^{K,L}C_L\\
 M_3^{K,L}C_L&0&-M_1^{K,L}C_L\\
 -M_2^{K,L}C_L&M_1^{K,L}C_L&0
 \end{bmatrix}.
\end{split}
\end{equation} 
\end{theorem}

\begin{proof}
Let \(Q^{(L)}=\Phi_XR_L\), whose degree-\(L\) projection is
\(Q_L=\Phi_{X,L}R_L\).  Each \(\mathcal L_i\) preserves harmonic degree, whereas
multiplication by \(\omega_q\), a degree-one spherical harmonic, couples
only degrees \(\ell\) and \(\ell\pm1\).  Hence
\[
  \Pi_K\mathcal L_iQ^{(L)}=\mathcal L_i\Pi_KQ_L,
  \qquad
  \Pi_KM_qQ^{(L)}=\Pi_KM_qQ_L
\]
whenever \(L\ge K+1\).  Projecting the exact continuous identity
\eqref{eq:Q-intertwining}, applied with right factor \(R_L\), to
\(\cH_K\) gives
\eqref{eq:guard-shell-identity}.
\end{proof}

The guard degree is essential to this identity. Multiplication can
map degree \(K+1\) to degree \(K\), so replacing the design by
\(\Pi_KM_q\Pi_KQ\) omits a retained contribution.

In the usual complex spherical-harmonic basis, the sparse
angular-momentum matrices are given by
\begin{align}
  \mathcal L_3Y_\ell^m&=mY_\ell^m,
  \label{eq:L3-harmonics}\\
  \mathcal L_\pm Y_\ell^m
  &=\sqrt{\ell(\ell+1)-m(m\pm1)}\,
    Y_\ell^{m\pm1},
  \label{eq:Lpm-harmonics}
\end{align}
where
\(
\mathcal L_\pm=\mathcal L_1\pm i\mathcal L_2
\).
These are the standard angular-momentum formulas; see
\cite{varshalovich1988} for conventions and coupling identities.
The multiplication matrices are Gaunt, equivalently
Clebsch--Gordan, couplings:
\begin{equation}
\label{eq:Gaunt}
  (M_q^{K,L})_{\ell m,\ell'm'}
  =\int_{\Sph^2}\overline{Y_\ell^m(\bomega)}\,
    \omega_qY_{\ell'}^{m'}(\bomega)\frac{\dd\bomega}{4\pi},
\end{equation}
and vanish unless \(\ell'=\ell\pm1\).

Write \(n_K=(K+1)^2\).  Then
\[
 \mathbb B_{K,L}\in\C^{3n_K\times3s},\qquad
  Y_K\in\C^{3n_K\times s},\qquad
 \Psi_L\in\C^{3s\times s}.
\]
\begin{corollary}
\label{cor:finite-exact-recovery}
Assume that \(\Phi_{X,L}\) has full column rank, \(L\ge K+1\), and
\(\rank\mathbb B_{K,L}=3s\).  Then
\begin{equation}
\label{eq:finite-exact-solve}
 \mathbb B_{K,L}^{\dagger} Y_K
 =\begin{bmatrix}\Psi_{1,L}\\\Psi_{2,L}\\\Psi_{3,L}\end{bmatrix},
 \qquad \Psi_{r,L}=R_L^{-1}D_r(X)R_L.
\end{equation}
%The joint spectral procedure in \cref{sec:joint-spectrum} recovers the point cloud from these three matrices.
\end{corollary}
\begin{proof}
Apply the left inverse \(\mathbb B_{K,L}^{\dagger}\) to
\eqref{eq:guard-shell-identity}, then use simultaneous similarity to
the coordinate diagonals.
\end{proof}

The row count \((K+1)^2\ge s\) is necessary for full block rank.
A generic sufficient cutoff is given in
Section~\ref{sec:finite-rank-sufficiency}.

For an arbitrary coefficient matrix \(C\in\C^{(L+1)^2\times s}\), define
\(\mathbb B_{K,L}[C]\) and \( Y_K[C]\) by
\eqref{eq:finite-data-design}, replacing \(C_L\) with \(C\) and
\(C_K\) with \(S_{K,L}C\). 
\begin{theorem}
\label{thm:GL-block-equivariance}
Suppose \(\mathbb B_{K,L}[C]\) has full column rank, and let
\[
 \widehat Z(C)=\mathbb B_{K,L}[C]^{\dagger} Y_K[C]
 =\argmin_Z\| Y_K[C]-\mathbb B_{K,L}[C]Z\|_{\rm F}^2.
\]
For every \(S\in GL(s,\C)\),
\begin{equation}
\label{eq:GL-block-equivariance}
 \widehat Z_r(CS)=S^{-1}\widehat Z_r(C)S,
 \qquad r=1,2,3.
\end{equation}
Consequently the pencil spectra are independent of the chosen basis of the empirical subspace.
\end{theorem}
\begin{proof}
Put \(\mathscr S=\diag(S,S,S)\).  The blocks satisfy
\[
 \mathbb B_{K,L}[CS]=\mathbb B_{K,L}[C]\mathscr S,
 \qquad  Y_K[CS]= Y_K[C]S.
\]
Since the design has full column rank,
\[
 (\mathbb B_{K,L}[C]\mathscr S)^{\dagger}
 =\mathscr S^{-1}\mathbb B_{K,L}[C]^{\dagger}.
\]
Hence \(\widehat Z(CS)=\mathscr S^{-1}\widehat Z(C)S\), which is
\eqref{eq:GL-block-equivariance}.  Under this similarity a right
eigenvector transforms by \(S^{-1}\) and a left eigenvector by \(S^*\),
so the coordinate quotients are unchanged.
\end{proof}

\subsection{Finite block identifiability and the generic cutoff}
\label{sec:finite-rank-sufficiency}

Let $n_K=(K+1)^2$. The guarded design has $3n_K$ rows and $3s$
columns, so full column rank requires
\begin{equation}
\label{eq:finite-block-necessary-count}
 n_K\ge s
\end{equation}
for $\rank\mathbb B_{K,L}=3s$. This count is not sufficient:
the cross-product structure may still produce a nontrivial kernel.
The two-multiplier criterion below is stronger. It requires
$\rank[A_1\ A_2]=2s$, and hence
\begin{equation}
\label{eq:finite-pair-necessary-count}
 n_K\ge2s.
\end{equation}
This count alone is also insufficient, since the multiplier blocks
may be singular or have intersecting ranges.
Theorem~\ref{thm:rank-generic-cutoff} supplies a generic sufficient
condition: $K\ge s$ implies full pair rank on a dense open set
with a measure-zero complement.

Thus $n_K\ge s$ is necessary for block rank, $n_K\ge2s$ is necessary
for the sufficient pair criterion, and $K\ge s$ guarantees that
criterion generically. The sharp cutoff may depend on $\kappa$,
the configuration, and the multiplier pair. Numerical conditioning
must be assessed from $\sigma_{\min}(\mathbb B_{K,L})$ or from
\eqref{eq:rank-pair-condition} and \eqref{eq:rank-angle-sufficient},
not from the dimension count alone.

%\subsubsection{A sufficient condition involving two multipliers}

Let \(L\ge K+1\), and let
\(C_L\in\C^{(L+1)^2\times s}\) be the harmonic coefficient matrix of a
signal basis.  Define
\begin{equation}
\label{eq:rank-multiplier-blocks}
 A_q=M_q^{K,L}C_L\in\C^{n_K\times s},\qquad q=1,2,3.
\end{equation}
The guarded design for one output column is
\begin{equation}
\label{eq:rank-guarded-design}
 \mathbb B_{K,L}=
 \begin{bmatrix}
  0&-A_3&A_2\\
  A_3&0&-A_1\\
  -A_2&A_1&0
 \end{bmatrix}\in\C^{3n_K\times3s}.
\end{equation}

\begin{proposition}
\label{prop:rank-two-multipliers}
If
\begin{equation}
\label{eq:rank-pair-condition}
 \rank[A_1\ A_2]=2s,
\end{equation}
then \(\rank\mathbb B_{K,L}=3s\).  The same statement holds with any
other pair of coordinate-multiplication blocks.
\end{proposition}

\begin{proof}
Suppose \(\mathbb B_{K,L}(z_1,z_2,z_3)^T=0\), where
\(z_r\in\C^s\).  The third block equation is
\[
 A_1z_2-A_2z_1
 =[A_1\ A_2]\begin{bmatrix}z_2\\-z_1\end{bmatrix}=0.
\]
Condition \eqref{eq:rank-pair-condition} gives \(z_1=z_2=0\).
The remaining equations imply \(A_1z_3=A_2z_3=0\).  Each \(A_q\),
\(q=1,2\), has full column rank under \eqref{eq:rank-pair-condition},
so \(z_3=0\).  Thus the kernel is trivial.  Permuting the coordinate
indices gives the other cases.
\end{proof}

Condition \eqref{eq:rank-pair-condition} is equivalent to injectivity of
both \(A_1\) and \(A_2\), together with
\begin{equation}
\label{eq:rank-range-intersection}
 \Range A_1\cap\Range A_2=\{0\}.
\end{equation}
It requires \(n_K\ge2s\), but that stronger dimension count alone is
still insufficient.  The condition is sufficient, not necessary, for
full rank of the three-coordinate block.  It is invariant under an
invertible change of signal basis, because replacing \(C_L\) by
\(C_LS\) right-multiplies \([A_1\ A_2]\) by \(\diag(S,S)\).

\begin{corollary}
\label{cor:rank-normalized-gram}
Suppose
\[
 G_1=A_1^*A_1\succ0,\qquad G_2=A_2^*A_2\succ0,
\]
and put
\begin{equation}
\label{eq:rank-angle-certificate}
 E=G_1^{-1/2}A_1^*A_2G_2^{-1/2},\qquad \rho=\|E\|_2.
\end{equation}
Then
\begin{equation}
\label{eq:rank-pair-singular-bound}
 \sigma_{\min}([A_1\ A_2])^2
 \ge(1-\rho)\min\{\lambda_{\min}(G_1),
                         \lambda_{\min}(G_2)\}
\end{equation}
and hence by \cref{prop:rank-two-multipliers} 
\begin{equation}
\label{eq:rank-angle-sufficient}
 \rho<1\quad\Longrightarrow\quad\rank\mathbb B_{K,L}=3s.
\end{equation}

\end{corollary}

\begin{proof}
The normalized concatenation
\[
 U=[A_1G_1^{-1/2}\ A_2G_2^{-1/2}]
\]
has Gram matrix
\[
 U^*U=\begin{bmatrix}I_s&E\\E^*&I_s\end{bmatrix}
 \succeq(1-\rho)I_{2s}.
\]
Since \([A_1\ A_2]=U\diag(G_1^{1/2},G_2^{1/2})\), this proves
\eqref{eq:rank-pair-singular-bound}.  If \(\rho<1\), the pair has full
column rank, and \cref{prop:rank-two-multipliers} applies.
\end{proof}

For a point configuration \(X=(\bx_1,\ldots,\bx_s)\), define the raw
atom multiplier maps
\begin{equation}
\label{eq:rank-raw-multipliers}
 N_{q,K}(X)=\Pi_KM_q\Phi_X:\C^s\longrightarrow\cH_K.
 \qquad q=1,2,3,
\end{equation}
We identify these maps with their harmonic coefficient matrices when
discussing matrix rank.  Multiplication raises harmonic degree by at most
one, so
\begin{equation}
\label{eq:rank-raw-guard}
 N_{q,K}(X)=\Pi_KM_q\Pi_L\Phi_X,
 \qquad L\ge K+1.
\end{equation}
If \(Q_L=\Pi_L\Phi_XR_L\), with \(R_L\) invertible, then
\begin{equation}
\label{eq:rank-basis-factorization}
 [A_1\ A_2]
 =[N_{1,K}(X)\ N_{2,K}(X)]\diag(R_L,R_L).
\end{equation}
Consequently the pair-rank condition is independent of the chosen signal
basis and, for the exact model, of \(L\) once \(L\ge K+1\).

There is no need to assume scalar full rank before checking the raw
condition.  Indeed, if \([N_{1,K}(X)\ N_{2,K}(X)]\) has rank \(2s\)
and \(\Pi_L\Phi_X\bc=0\), then \eqref{eq:rank-raw-guard} gives
\(N_{1,K}(X)\bc=0\).  Injectivity of \(N_{1,K}(X)\) implies
\(\bc=0\).  Thus the raw pair condition itself proves that
\(\Phi_{X,L}=\Pi_L\Phi_X\) has full column rank and that an invertible
orthonormalizing factor \(R_L\) exists.

%\subsubsection{A generic sufficient cutoff}

The following theorem gives an a priori sufficient degree for generic
configurations. 
\begin{theorem}
\label{thm:rank-generic-cutoff}
Fix \(\kappa>0\), an integer \(s\ge1\), and \(K\ge s\).
In the ordered configuration space
\[
 \Omega_s=\{(\bx_1,\ldots,\bx_s)\in B_1(0)^s:
                   \bx_j\ne\bx_k\text{ for }j\ne k\},
\]
there is a dense open subset whose complement has Lebesgue measure zero
on which
\begin{equation}
\label{eq:rank-generic-pair}
 \rank[N_{1,K}(X)\ N_{2,K}(X)]=2s.
\end{equation}
For every configuration in this subset and every \(L\ge K+1\), the
truncated scalar synthesis has rank \(s\), and the guarded ESPRIT block
formed from any basis of its range has rank \(3s\).
\end{theorem}

The proof is given in \cref{app:generic-cutoff}.

%\subsubsection{Exceptional configurations and numerical interpretation}

The generic cutoff does not exclude exceptional configurations.
Consider \(s=1\), \(K=1\), and a positive zero \(t\)
of \(j_1\).  Choose \(\kappa>t\), so that
\(\bx=(t/\kappa)\be_3\in B_1(0)\), and write
\(f(\bomega)=e^{it\omega_3}\).  Direct projection onto degrees zero
and one gives
\begin{align}
 \Pi_1(\omega_1f)&=3\frac{j_1(t)}t\omega_1=0,
 \label{eq:rank-resonance-first}\\
 \Pi_1(\omega_2f)&=3\frac{j_1(t)}t\omega_2=0,
 \label{eq:rank-resonance-second}\\
 \Pi_1(\omega_3f)
 &=ij_1(t)+3\left(j_0(t)-\frac{2j_1(t)}t\right)\omega_3
 =3j_0(t)\omega_3\ne0.
 \label{eq:rank-resonance-third}
\end{align}
Indeed, \(j_1(t)=(\sin t-t\cos t)/t^2=0\) implies
\(j_0(t)=\cos t\ne0\).  The scalar synthesis through any degree
\(L\ge2\) consequently has rank one.  Nevertheless, its normalized
multiplier blocks satisfy \(A_1=A_2=0\) and \(A_3\ne0\), so
\begin{equation}
\label{eq:rank-resonance-block}
 \rank\begin{bmatrix}0&-A_3&0\\A_3&0&0\\0&0&0\end{bmatrix}=2<3.
\end{equation}
Thus even \(K\ge s\) and scalar full rank do not guarantee guarded-block
rank for every configuration and frequency.

The finite coefficient matrix can be tested by
\eqref{eq:rank-pair-condition} or \eqref{eq:rank-angle-sufficient}.
For exact Fourier atoms, Theorem~\ref{thm:rank-generic-cutoff}
gives a sufficient cutoff outside an exceptional set.
Configurations near that set or near node collisions may nevertheless
be poorly conditioned.

A perturbation estimate requires alignment of the empirical frame
with an exact orthonormal frame $Q$. Such a change of basis preserves
the exact joint spectrum. Choose a Procrustes factor
\begin{equation}
\label{eq:frame-Procrustes-alignment}
 U_\star\in\argmin_{U\in\mathrm \cU(s)}
 \|\widehat Q-QU\|_{L_{\rm F}^2},
\end{equation}
where \(\cU(s)\) is the unitary group of degree \(s\).

The following bound compares the exact and empirical coordinate
stacks after their degree-$L$ frames have been aligned. The matrix
$\Psi_L$ is expressed in that aligned exact basis.
\begin{proposition}
\label{prop:finite-perturbation}
Let $C_L$ be the coefficient matrix of an exact frame satisfying
\[
Y_K[C_L]=\mathbb B_{K,L}[C_L]\Psi_L.
\qquad L\ge K+1,
\]
Let $\widehat C_L=C_L+E$ be the aligned empirical coefficient matrix.
If $$\|E\|_2<\sigma_{\min}(\mathbb B_{K,L}[C_L])$$ then
$\mathbb B_{K,L}[\widehat C_L]$ has full column rank, and
\[
\widehat\Psi_L
:=
\mathbb B_{K,L}[\widehat C_L]^\dagger
Y_K[\widehat C_L]
\]
satisfies
\[
\|\widehat\Psi_L-\Psi_L\|_{\mathrm F}
\le
\frac{
\kappa^{-1}\sqrt{K(K+1)}\,\|E\|_{\rm F}
+\|E\|_2\|\Psi_L\|_{\rm F}
}{
\sigma_{\min}(\mathbb B_{K,L}[C_L])-\|E\|_2
}.
\]
\end{proposition}

\begin{proof}
Pointwise cross multiplication by $\omega$ and orthogonal
harmonic projection are contractions. Consequently,
\[
\|\mathbb B_{K,L}[E]\|_2\le \|E\|_2. 
\]
The singular-value perturbation inequality therefore gives
\[
\sigma_{\min}(\mathbb B_{K,L}[\widehat C_L])
\ge \sigma_{\min}(\mathbb B_{K,L}[C_L])-\|E\|_2>0.
\]
The angular-momentum identity on harmonics of degree at most $K$
also yields
\[
\|Y_K[E]\|_{\rm F}
\le \kappa^{-1}\sqrt{K(K+1)}\ \|E\|_{\rm F}.
\]
By linearity of the data and design maps,
\[
\widehat\Psi_L-\Psi_L
=
\mathbb B_{K,L}[\widehat C_L]^\dagger
\bigl(Y_K[E]-\mathbb B_{K,L}[E]\Psi_L\bigr).
\]
Taking Frobenius norms proves the estimate.
\end{proof}

At a fixed cutoff, the estimate is locally Lipschitz in the aligned
coefficient frame as long as its perturbation is smaller than the
least singular value of the exact block. The differentiation term contains
$\sqrt{K(K+1)}/\kappa$; the design-perturbation term need not decay
with increasing wavenumber.

\section{Algorithms}
\label{sec:algorithms-perturbations}

Both algorithms use one calibrated empirical signal frame
$\widehat Q$. Such a frame can be obtained from snapshots
\begin{equation}
 f_n(\bomega)=\sum_{j=1}^s a_{jn}\varphi_{\bx_j}(\bomega),
 \qquad n=1,\ldots,N_{\rm snap}.
 \label{eq:snapshots}
\end{equation}
Here the amplitude matrix $A=(a_{jn})$ has row rank $s$. 

For perturbation analysis, replace $Q$ by $QU_\star$, with
$U_\star$ as in \eqref{eq:frame-Procrustes-alignment}, and conjugate
the exact coordinate and shift matrices by $U_\star^*(\cdot)U_\star$.
In this aligned basis, write
\begin{equation}
\label{eq:aligned-frame-error}
 \widehat Q=Q+E.
\end{equation}
If
$\mathcal B_{\widehat Q}$ has full column rank and
$\widehat\Psi=\kappa^{-1}\mathcal B_{\widehat Q}^\dagger
\mathcal L\widehat Q$, then
\[
\widehat\Psi-\Psi
=
\mathcal B_{\widehat Q}^{\dagger}
\left(\kappa^{-1}\mathcal LE-\mathcal B_E\Psi\right).
\]
For a paired-circle relation $Q_+=Q_-F$, write
$\widehat Q_\pm=Q_\pm+E_\pm$. If $\widehat Q_-$ has full column rank,
$\widehat F=\widehat Q_-^\dagger\widehat Q_+$ satisfies
\[
\widehat F-F
=
\widehat Q_-^\dagger(E_+-E_-F).
\]
The generator error involves angular derivatives of $E$, whereas
the finite-shift error involves its paired restrictions. The alignment
is needed to compare coefficient matrices in the same basis.

\subsection{Practical paired-circle algorithm}

\begin{enumerate}[leftmargin=*,label=\textbf{Step \arabic*.}]
\item 
Estimate the signal rank and obtain one basis $\widehat Q$ whose
restrictions can be evaluated consistently at all paired points.

\item 
Use three orthonormal directions or a well-conditioned spanning family.
Choose $0<\eps_0<\min\{\pi/\kappa,2\}$ in the unit-ball model, set
$\bh_{0,r}=\eps_0\bv_r$, and put $q_0=\kappa \eps_0$.

\item 
For every direction compute
$\widehat \Phi_{\eps_0\bv_r}=\widehat Q_{r,-}^{\dagger}\widehat Q_{r,+}, r=1,2,3$.
Test a prescribed finite list of normalized pencil coefficient
vectors. One rule, independent of the unknown targets, is to diagonalize
$P_m=\sum_r\alpha_r^{(m)}\widehat \Phi_{\eps_0\bv_r}
=V_m\Lambda_mV_m^{-1}$, normalize the columns of $V_m$ to unit
Euclidean norm, and maximize
\begin{equation}
\label{eq:common-pencil-score}
 S_m=
 \frac{\min_{j\ne k}|\lambda_j^{(m)}-\lambda_k^{(m)}|}
 {\max\{\|P_m\|_2,\epsilon_{\rm mach}\}\,\cond(V_m)}.
\end{equation}
Reject candidates whose eigenvalue separation is unresolved, and
choose the first maximizing index in a tie. Any alternative rule
must specify the coefficient distribution, normalization, rejection
criterion, and tie-break. Let $V$ be the selected basis and use
the diagonal entries of $V^{-1}\widehat F V$, equivalently the
two-sided estimates.

\item 
At each scalar chord length $0<\eps_m<2$, put $q_m=\kappa \eps_m$ and
$\bh_{m,r}=\eps_m\bv_r$. Compute the paired matrices and read all diagonal
phases in the same basis $V$. For each target and direction, let
$\theta_m$ be the current principal phase and $\widehat t_{m-1}$ the
previous projection estimate. Initialize $\widehat t_0=\theta_0/q_0$
at the baseline shift, and for $m\ge1$ set
\begin{equation}
 n_m=\round\!\left(\frac{q_m\widehat t_{m-1}-\theta_m}{2\pi}\right),
 \qquad \widehat t_m=\frac{\theta_m+2\pi n_m}{q_m}.
 \label{eq:cmp-unwrap}
\end{equation}
Increase shifts gradually and check restriction conditioning at every
stage (see the discussion in \cref{sec:long-shift}). 

\item 
After unwrapping, let $\widehat t_{j,r}$ be the estimated projection
of target $j$ onto direction $\bv_r$. For three orthonormal directions,
set $\widehat\bx_j=\sum_{r=1}^3\widehat t_{j,r}\bv_r$.
For a larger spanning family, find $\widehat\bx_j$ by least squares from
the equations $\bv_r\cdot\widehat\bx_j\approx\widehat t_{j,r}$.

\end{enumerate}

\subsection{Practical whole-sphere generator algorithm}
\label{sec:algorithm}

The input is an estimated $s$-dimensional subspace of $\cH_L$.
If the snapshot coefficient matrix \(A\) in \eqref{eq:snapshots}
has row rank \(s\) and \(\Phi_{X,L}\) has column rank \(s\),
the exact projected snapshots span \(\Range\Phi_{X,L}\).
An SVD of their harmonic coefficients gives an orthonormal basis.
For noisy snapshots, the same procedure gives the empirical coefficient
matrix \(\widetilde C_L\in\C^{(L+1)^2\times s}\).
Any other subspace-estimation or lifting procedure that supplies
such a matrix can be used.

\begin{enumerate}[leftmargin=*,label=\textbf{Step \arabic*.}]
\item
Obtain \(\widetilde C_L\in\C^{(L+1)^2\times s}\) with
\(\widetilde C_L^*\widetilde C_L=I_s\).  

\item 
Choose $K=\max\{s,K_{\rm energy}\}$ and $L\ge K+1$ when
using the generic sufficient cutoff of \cref{thm:rank-generic-cutoff}.
Here $K_{\rm energy}$ resolves the useful signal and noise bandwidth.
Set $\widetilde C_K=S_{K,L}\widetilde C_L$ and retain degree $K+1$
for multiplication before output projection. 
\item 
Form
\begin{equation}
\label{eq:empirical-generator-data}
 \widetilde Y_K=\frac1\kappa
 \begin{bmatrix}
 L_1^{(K)}\widetilde C_K\\
 L_2^{(K)}\widetilde C_K\\
 L_3^{(K)}\widetilde C_K
 \end{bmatrix},
 \qquad T_q=M_q^{K,L}\widetilde C_L,
 \quad q=1,2,3.
\end{equation}
Then assemble
\begin{equation}
\label{eq:empirical-block}
 \widetilde{\mathbb B}_{K,L}=
 \begin{bmatrix}
 0&-T_3&T_2\\ T_3&0&-T_1\\-T_2&T_1&0
 \end{bmatrix}.
\end{equation}

\item 
Compute a thin SVD or a rank-revealing QR factorization of
\(\widetilde{\mathbb B}_{K,L}\); see \cite{golubvanloan2013}.
For a full-rank solve, set
\begin{equation}
\label{eq:practical-LS}
 \widetilde\Psi=
 \begin{bmatrix}\widetilde\Psi_1\\\widetilde\Psi_2\\\widetilde\Psi_3\end{bmatrix}
 =\widetilde{\mathbb B}_{K,L}^{\dagger}\widetilde Y_K.
\end{equation}
If the required rank is not resolved, reconsider
the retained degree, the estimated signal rank, or the available data.

The normalized commutators
\begin{equation}
\label{eq:commutator-diagnostic}
 c_{rt}=
 \frac{\|[\widetilde\Psi_r,\widetilde\Psi_t]\|_{\rm F}}
 {\max\{\|\widetilde\Psi_r\|_{\rm F}\|\widetilde\Psi_t\|_{\rm F},
          \epsilon_{\rm mach}\}},\qquad r<t,
\end{equation}
vanish for the exact full-rank model. Here
\(\epsilon_{\rm mach}\) prevents division by zero.
These quantities provide numerical consistency checks
(\cref{sec:compatibility}).

\item Apply the prespecified common-pencil rule, such as
\eqref{eq:common-pencil-score}, to normalized combinations of
$\widetilde\Psi_1,\widetilde\Psi_2,\widetilde\Psi_3$. Let \(v_j,w_j\) be right and left eigenvectors of the selected pencil
\[
\widetilde \Psi_\alpha=\sum_{r=1}^3\alpha_r\widetilde\Psi_r.
\] Let \begin{equation}
\label{eq:noisy-left-right-readout}
z_{j,r}=\frac{w_j^*\widetilde\Psi_rv_j}{w_j^*v_j},\quad r=1,2,3,\quad j=1,\dots, s.
\end{equation}
Take
\begin{equation}
\label{eq:noisy-left-right-readout2}
\widehat x_{j,r}=\operatorname{Re}(z_{j,r}),
\end{equation}
as the original coordinate estimate. The two-sided formula applies
to nonnormal pencils \cite[Section~4]{hekressnerplestenjak2025}.
For exact eigenvectors of the selected pencil,

\begin{equation}
\label{eq:projected-readout-check}
\sum_r\alpha_r z_{j,r}=\lambda_j(\widetilde\Psi_\alpha)
\end{equation}
so a discrepancy measures numerical inconsistency in this calculation.

When a smooth empirical aperture frame in the same coefficient basis
is available, \cref{sec:phase-readout} gives an alternative coordinate
estimate from these same right eigenvectors.
\end{enumerate}

For directions \(\bomega_n\) with positive weights \(w_n\), set
\(W_h=\diag(w_1,\ldots,w_N)\otimes I_3\), as in
\cref{rem:discrete-equivalence}.  Let \(\widetilde Q\) be one common
smooth interpolant of the empirical signal basis.  Form node-major blocks
\[
 \widetilde{\cB}_hZ
 =\operatorname{stack}_n\left(
 \bomega_n\times\sum_{r=1}^3\be_r
       \widetilde Q(\bomega_n)Z_r\right),
 \qquad
 \widetilde Y_h
 =\operatorname{stack}_n\frac1\kappa
    (\boldsymbol{\mathcal L}\widetilde Q)(\bomega_n),
\]
and solve
\begin{equation}
\label{eq:nodal-weighted-LS}
 \min_Z\|W_h^{1/2}
   (\widetilde Y_h-\widetilde{\cB}_hZ)\|_{\rm F}^2.
\end{equation}

The nodal system is not the guarded harmonic system. To realize
the latter at nodes, first form the degree-\(K\) data and design
in \eqref{eq:finite-data-design}, then evaluate those projected
fields. A quadrature exact through degree \(2K\) reproduces the
harmonic Frobenius objective, since products of retained components
have degree at most \(2K\). Projection of the tangential fields
followed by the Hodge map generally gives a different operator:
componentwise harmonic projection need not commute with multiplication
by \(\bomega\).

\subsection{MUSIC initialization and refinement}
\label{sec:music-initialization}
We use a common MUSIC objective to compare the ESPRIT initializers
\cite{schmidt1986,fannjiangnguyen2026}. The relevant issues are the
size of the initial error relative to a MUSIC well, coverage of
distinct targets, and the errors remaining after local refinement.

For a single noiseless atom at $\bx$, the MUSIC objective is
\begin{equation}
 J_{\bx}(\by)=1-\sinc^2(\kappa|\by-\bx|).
 \label{eq:music-spatial-scale}
\end{equation}
Its central well has spatial scale $1/\kappa$. An error that is small
in absolute units can therefore become large relative to that well
as frequency increases.

The off-model mixture $f=A\varphi_\bx+B\varphi_\bz$, with complex
coefficients and $\bz\ne\bx$, illustrates why coordinate readout
matters. When weighted cross correlations in the normal equations are
negligible, the original generator estimate is approximately
\[
 \by_g\simeq\frac{|A|^2\bx+|B|^2\bz}{|A|^2+|B|^2}.
\]
This displacement need not shrink with $\kappa$ and can exceed the
narrowing MUSIC scale. Reliable long paired shifts and phase recovery
after demixing can produce errors on that scale under their respective
hypotheses (\cref{app:scalar-kernel,app:phase-normalization}). 

All reported initializations are refined against the same empirical
objective,
\begin{equation}
 J_{\widehat Q}(\by)=1-\norm{\widehat Q^*\varphi_\by}^2,
 \qquad \widehat Q^*\widehat Q=I.
 \label{eq:cmp-music}
\end{equation}
Accurate proposals near one target do not compensate for a missing
target. Conversely, if two initializations enter the same attraction
regions, refinement may remove most of their initial accuracy
difference, although their iteration counts may differ. We therefore
report both raw error and distinct-target coverage.

\subsection{Phase-gradient recovery after demixing}
\label{sec:phase-readout}

The common eigenbasis may separate approximate atoms even when
the coordinates in \eqref{eq:noisy-left-right-readout}--\eqref{eq:noisy-left-right-readout2}
are biased. Let $V=[\bv_1\ \cdots\ \bv_s]$ be the selected generator
eigenbasis. For the exact continuous model, $V=R^{-1}\Lambda$ up to permutation,
with $\Lambda$ nonsingular diagonal, so $Q\bv_j=a_j\varphi_{\bx_j}$.

Let $F$ be a smooth empirical frame in the same coefficient
basis. Form the demixed aperture functions
\begin{equation}
 f_j(\bomega)=F(\bomega)\bv_j.
 \label{eq:phase-demixing}
\end{equation}
Amplitude normalization preserves the phase of $f_j$.
For a nonvanishing demixed function, define the phase-gradient estimate
\begin{equation}
 \bx_j^{\,\rm ph}
 =\frac{3}{2\kappa}\int_{\Sph^2}
 \operatorname{Im}\left(
       \frac{\nabla_{\Sph^2}f_j}{f_j}\right)\frac{d\bomega}{4\pi}.
 \label{eq:phase-readout-main}
\end{equation}
The estimate is invariant under nonzero rescaling of $\bv_j$ and is exact
for a pure atom. 

If, for a permutation $\pi$,
\[
 f_j=a_j\varphi_{\bx_{\pi(j)}}(1+\xi_j),
 \qquad a_j\ne0,\qquad
 \xi_j\in C^1(\Sph^2;\C),\quad \|\xi_j\|_\infty\le\rho_j<1,
\]
then \cref{thm:hybrid-pointwise-dominance} gives
\[
 |\bx_j^{\rm ph}-\bx_{\pi(j)}|
 \le\frac{3}{2\kappa}\arcsin\rho_j.
\]
\Cref{app:phase-normalization} gives the proof, a sufficient coefficient
condition, the effect of uniform reconstruction error, and the sharper
two-atom bound. 

To implement the phase-gradient estimate, specify a smooth representation
of $F$ and check convergence of the spherical quadrature. If only
harmonic coefficients are available, let $\widehat C_L$ denote the
observed frame through the guard degree and define
\begin{equation}
 f_{j,L}(\bomega)
 =\sum_{\ell=0}^{L}\sum_{m=-\ell}^{\ell}
       (\widehat C_L\bv_j)_{\ell m}Y_{\ell m}(\bomega).
 \label{eq:phase-harmonic-reconstruction}
\end{equation}
Both tangential derivatives can be synthesized from the same
coefficients by the angular-momentum recurrences. This is exact for the
finite expansion, but it does not make the quotient
$\nabla_{\Sph^2}f_{j,L}/f_{j,L}$ bandlimited. Harmonic truncation and
spherical quadrature are therefore separate approximation mechanisms.

\section{Numerical experiments}
\label{sec:matched-regimes}

The comparison uses matched synthetic signal and noise realizations for
both initializers and the same subsequent MUSIC iteration. The harmonic
cutoff of the whole-sphere realization in the primary comparison
satisfies $K\ge s$. In the baseline comparisons,
``whole sphere'' denotes the original diagonal coordinate estimate
\eqref{eq:noisy-left-right-readout}-\eqref{eq:noisy-left-right-readout2}. We preserve those results and
then test the phase-normalized estimate separately in
\cref{sec:phase-readout-experiment}, holding the generator matrices
and selected common eigenbasis fixed.

Here ``matched'' means that the two methods share the same target cloud,
auxiliary centers, mixing coefficients, and nominal perturbation level.
The baseline experiments compare the two specified
native implementations. The additional benchmark in
\cref{sec:common-harmonic-benchmark} uses identical retained harmonic
data throughout. 
\subsection{Experimental protocol}

The ten target clouds are fixed across wavenumbers and noise models.
Each contains $s=125$ points obtained from a centered $5\times5\times5$
Cartesian scaffold of spacing $0.2$. Independent jitters are uniform in
volume in the shell with radii $0.05$ and $0.10$; their sample mean is
removed before applying a Haar-distributed rotation and a translation
uniform in the ball of radius $0.02$. The resulting clouds lie in the
unit ball. Their minimum separations $\delta_X$ range from $0.0325123$ to
$0.0731591$, with median $0.0592493$. The wavenumbers are
$\kappa=10,20,40,80,160,320,640,1280$. These parameters are outside of 
the certified regime for MUSIC in \cite{fannjiangnguyen2026}, which
requires $\kappa\delta_X\ge\frac92s^{2/3}=112.5$ at the minimum.

This deliberately structured ensemble permits a controlled comparison
over frequency, but it is not a representative distribution over all
point clouds. In particular, it does not test strongly clustered,
adversarially projected, or broadly unstructured configurations. The
paired-circle implementation also fixes the three Cartesian axes rather
than optimizing or randomizing them for each cloud. Its reported
restriction failures therefore apply to this axis choice and should not
be interpreted as failures of every paired-circle design.

Both noise models use 250 auxiliary centers sampled uniformly in the
ball of radius $0.9$, rejecting a candidate within distance $0.04$ of
any target or previously accepted auxiliary center. A $250\times125$
complex Gaussian matrix, with independent entries
$(g_1+ig_2)/\sqrt2$ and $g_1,g_2\sim N(0,1)$, is orthonormalized to
supply the auxiliary mixing matrix $M$. The same auxiliary centers and
mixing are used across frequencies and perturbation levels for each cloud.
Let $H_X$ and $H_A$ denote the target and auxiliary harmonic dictionaries
through degree $L$. Additive perturbations use 125 snapshots
\[
 S_\varepsilon=H_X+\varepsilon\alpha H_AM,
 \qquad \alpha=\frac{\|H_X\|_{\rm F}}{\|H_AM\|_{\rm F}},
\]
followed by orthonormalization of their range. Thus the stated additive
percentage is the relative Frobenius perturbation before subspace
extraction, not the resulting subspace angle.

For equal-angle perturbations, an orthonormal frame $Q$ for
$\Range H_X$ and a frame $W$ obtained by projecting $H_AM$ onto its
orthogonal complement and orthonormalizing satisfy
$Q^*W=0$ and $Q^*Q=W^*W=I$. We use
\[
 \widehat Q=\sqrt{1-\varepsilon^2}\,Q+\varepsilon W.
\]
Every principal sine is then $\varepsilon$: $0.01$ or $0.05$ in the
reported noisy cases. Orthogonality, Frobenius normalization, and
principal angles are defined in the finite harmonic representation
through degree $L$, not separately imposed in continuous
$L^2(\mathbb S^2)$. Both perturbation models use smooth Fourier
components from the specified auxiliary dictionary.
Thus these are structured, smooth, finite-dictionary aperture
perturbations. They are not independent sample noise, arbitrary
measurement error, or adversarial model mismatch. All conclusions below
are conditional on this perturbation construction.

For common-eigenbasis selection, the whole-sphere method tests
128 random real linear combinations of its coordinate matrices.
The paired-circle method tests 127 random complex linear combinations
of its baseline-shift matrices and one fixed combination with coefficients
$(1,\sqrt{2},\sqrt{3})$, and reuses the selected basis at subsequent shifts.
Paired circle estimates use 4000 samples per
circle. \commentout{
The different real and complex candidate laws are part of the reported
protocol and may influence spectral conditioning. A strict pencil-only
benchmark would use the same normalized coefficient law for both methods.

The stored numerical record fixes these candidate lists and seeds but
does not retain a complete candidate score, rejection rule, and tie-break
independently of the implementation. The candidate counts alone therefore
do not make the spectral-selection step reproducible from the manuscript.
For a new implementation, \eqref{eq:common-pencil-score} supplies one
fully specified truth-independent rule; an archival release should record
the exact rule used together with the candidate law and normalization.
}

Both initializers are followed by the same 200 unprojected MUSIC
gradient steps with step size $3/(2\kappa^2)$. The scale is motivated by
\cite[Theorem~2.2(v)]{fannjiangnguyen2026}, but that theorem's
separation, initial-basin, and perturbation hypotheses are not verified
for these runs, and the iterates are not projected onto certified wells.
MUSIC is therefore an empirical common postprocessor here, not a
theorem-backed recovery guarantee.

For two clouds of $s$ points, the reported error is the exact bottleneck
distance
\begin{equation}
 E_\infty(\widehat X,X)=
 \min_{\pi}\max_j|\widehat\bx_j-\bx_{\pi(j)}|
 \label{eq:cmp-bottleneck}
\end{equation}
where $\pi$ runs through the permutation group. 
Complete recovery means all 125 points admit a one-to-one matching
within $\rho=7/(12\kappa)$ (see \cite{fannjiangnguyen2026}, Corollary 2.3). This is an operational tolerance. A
numerical rank failure is recorded as no output. 

The harmonic cutoff uses the maximum radius $R_{\mathcal Z}$ of the
combined target and auxiliary centers $\cZ$. With $b=\kappa R_{\mathcal Z}$,
\begin{equation}
 K_{\rm energy}=\max\{15,\lceil b\rceil+
                     \max\{20,\lceil8b^{1/3}\rceil\}\},
 \qquad K=\max\{s,K_{\rm energy}\},\qquad L=K+1.
 \label{eq:cmp-cutoff-rule}
\end{equation}
This is the bandwidth rule used in the computations; $K\ge s$ enforces
the generic sufficient degree in \cref{thm:rank-generic-cutoff}.
The resulting degrees are listed in \cref{tab:cmp-cutoffs}.
\begin{table}[ht]
\centering\small
\begin{tabular}{rrrrr}
\toprule
$\kappa$ & $K_{\rm energy}$ & $K$ & $L$ & Input coefficients per column\\
\midrule
10 &29 &125 &126 &16129\\
20 &39 &125 &126 &16129\\
40 &63 &125 &126 &16129\\
80 &106&125 &126 &16129\\
160&186&186 &187 &35344\\
320&340--341&340--341&341--342&116964--117649\\
640&641--643&641--643&642--644&413449--416025\\
1280&1232--1236&1232--1236&1233--1237&1522756--1532644\\
\bottomrule
\end{tabular}
\caption{Retained degree $K$, guard degree $L$, and harmonic input size
$(L+1)^2$ per frame column. Ranges reflect the maximum dictionary radius
across the ten clouds.}
\label{tab:cmp-cutoffs}
\end{table}
The condition $K\ge s$ uses the generic rank theorem rather than a
cloud-dependent rank test. At $\kappa\le80$, it raises $K$ from the
energy cutoff $29$--$106$ to $125$. The additional modes can establish
exact rank while carrying little energy; their inclusion does not by
itself imply stable recovery.

The native representations also have substantially different sizes. At
$\kappa=1280$, the 125-column whole-sphere input contains between
$190{,}344{,}500$ and $191{,}580{,}500$ complex harmonic coefficients,
about $3.05$--$3.07$ GB in complex double precision before work arrays.
At one displacement, the Cartesian paired-circle implementation evaluates
$6\times4000=24{,}000$ aperture values per frame column, although it
repeats this operation over several displacements. These counts are not
a complete complexity comparison because the implementations stream and
reuse data differently. Wall-clock time, peak memory, and total arithmetic
were not recorded, so no computational-efficiency conclusion is drawn.

The primary comparison comprises noiseless, 1\% additive, and 1\%
equal-angle cases at all eight wavenumbers: 240 data sets and 480
method/case records. All 240 whole-sphere blocks have numerical rank
$3s=375$. The 5\% additive comparison uses
$\kappa=160,320,640,1280$; the additional 5\% equal-angle comparison
uses $\kappa=1280$.

\begin{table}[htbp]
\centering\small
\begin{tabular}{@{}p{0.24\textwidth}p{0.70\textwidth}@{}}
\toprule
Parameter & Specification\\
\midrule
Random seeds & Trial $j=0,\ldots,9$: geometry seed
$20260901+997j$; auxiliary dictionary and mixing seed equal to the
geometry seed plus $1000003$; pencil seed $314159+j$.
NumPy's \texttt{default\_rng} is used.\\
Circle geometry & Cartesian axes; 4000 equally spaced samples per
circle, with angles $2\pi(n+0.371)/4000$, $n=0,\ldots,3999$.\\
Phase shifts & Baseline $q_0=\kappa \eps_0=3$. Subsequent $q$ values are the sorted,
distinct members of
$\{50,500,\min(5000,0.75\kappa),1.5\kappa\}$ lying in $(3,2\kappa)$.
The aperture chord length is $q/\kappa$; the baseline eigenbasis is
reused for phase continuation.\\
Rank decisions & Direct generator and circle least-squares solves use
relative singular-value cutoff $10^{-12}$. At $\kappa=320,640,1280$,
the generator normal matrix is accumulated in degree blocks and accepted
only when its smallest-to-largest eigenvalue ratio exceeds $10^{-6}$.\\
Refinement & 200 unprojected MUSIC steps of size $3/(2\kappa^2)$;
no adaptive retries, shift retuning, clipping, or truth-based stopping.\\
\bottomrule
\end{tabular}
\caption{Reproducibility parameters common to the reported comparisons.
The signal rank is fixed at the known value $s=125$.}
\label{tab:cmp-protocol}
\end{table}

The baseline methods use the same synthetic frame but different
evaluations of it. Write $F=\Phi_{\mathcal Z}C$, where $\mathcal Z$
contains the target and auxiliary centers. The whole-sphere estimate
uses the normalized harmonic frame $\widehat Q_L=\Pi_LF$ and the
generator equations through degree $K$. For $\kappa>10$, the
paired-circle method evaluates $F$ directly on the circles, and MUSIC
uses the continuous sinc kernel with the same coefficients $C$.
Thus the circle restrictions are not computed from $\widehat Q_L$.
At $\kappa=10$, harmonic evaluation is used for numerical stability,
with additional evaluation compression bounded uniformly by $10^{-12}$.

The baseline comparison therefore matches the underlying synthetic
realization, not the discrete input at every stage.
Section~\ref{sec:common-harmonic-benchmark} gives a separate comparison
in which the generator actions, circle restrictions, and MUSIC objective
all use one empirical harmonic frame. It also measures the defect in
the noiseless shift identity caused by this finite representation.

\subsection{Baseline matched recovery results}

\begin{table}[htbp]
\centering\small
\begin{tabular}{rcccccc}
\toprule
&\multicolumn{2}{c}{Noiseless}&\multicolumn{2}{c}{1\% additive}
&\multicolumn{2}{c}{1\% equal angle}\\
\cmidrule(lr){2-3}\cmidrule(lr){4-5}\cmidrule(lr){6-7}
$\kappa$&Whole&Circle&Whole&Circle&Whole&Circle\\
\midrule
10 &10/10&0/10$^{\rm R}$&0/10&0/10$^{\rm R}$&0/10&0/10$^{\rm R}$\\
20 &10/10&0/10$^{\rm R}$&0/10&0/10$^{\rm R}$&0/10&0/10$^{\rm R}$\\
40 &10/10&0/10$^{\rm R}$&4/10&0/10$^{\rm R}$&3/10&0/10$^{\rm R}$\\
80 &10/10&0/10$^{\rm R}$&10/10&0/10&9/10&0/10\\
160&10/10&10/10&10/10&10/10&10/10&10/10\\
320&10/10&10/10&10/10&10/10&10/10&10/10\\
640&10/10&10/10&10/10&10/10&10/10&10/10\\
1280&10/10&10/10&10/10&10/10&10/10&10/10\\
\bottomrule
\end{tabular}
\caption{Baseline complete recovery after 200 identical MUSIC steps.
Whole sphere uses the original diagonal estimate. Superscript R denotes
ten numerical restriction-rank failures and hence no circle proposals. Other zero entries denote
outputs that fail complete-cloud matching.}%Source: \cite{esprit-comparison-data}.}
\label{tab:cmp-success}
\end{table}

At $\kappa=80$, the whole-sphere initializer gives more accurate
recovery than the paired-circle implementation. Under 1\% equal-angle noise the median
maximum circle restriction condition number is about $2.36\times10^{10}$,
whereas at $\kappa=160$ it is about $233$. Noise can raise the numerical rank of a restricted matrix without
making its inversion stable, as the noisy $\kappa=80$ circle outputs
illustrate. The corresponding median
whole-sphere block singular-value ratios are approximately $0.721$ and
$0.854$ in the equal-angle cases.

\begin{table}[htbp]
\centering\small
\begin{tabular}{rlrrr}
\toprule
$\kappa$&Initializer&Raw median $E_\infty$&After MUSIC&Complete\\
\midrule
10 &Whole sphere&$7.6498\times10^{-1}$&$3.6467\times10^{-1}$&0/10\\
20 &Whole sphere&$2.7858\times10^{-1}$&$2.9928\times10^{-1}$&0/10\\
40 &Whole sphere&$1.3839\times10^{-1}$&$1.6680\times10^{-1}$&3/10\\
80 &Whole sphere&$1.7044\times10^{-2}$&$1.0147\times10^{-5}$&9/10\\
80 &Paired circle&$5.1636\times10^{-1}$&$5.1977\times10^{-1}$&0/10\\
160&Whole sphere&$2.9542\times10^{-3}$&$2.4311\times10^{-6}$&10/10\\
160&Paired circle&$2.5342\times10^{-4}$&$2.4311\times10^{-6}$&10/10\\
320&Whole sphere&$7.3870\times10^{-4}$&$5.9843\times10^{-7}$&10/10\\
320&Paired circle&$8.5294\times10^{-6}$&$5.9843\times10^{-7}$&10/10\\
640&Whole sphere&$2.8103\times10^{-4}$&$1.4371\times10^{-7}$&10/10\\
640&Paired circle&$1.8663\times10^{-6}$&$1.4371\times10^{-7}$&10/10\\
1280&Whole sphere&$1.5793\times10^{-4}$&$3.8619\times10^{-8}$&10/10\\
1280&Paired circle&$6.4099\times10^{-7}$&$3.8619\times10^{-8}$&10/10\\
\bottomrule
\end{tabular}
\caption{Baseline localization under 1\% equal-angle random-complement
noise. Whole sphere uses the original diagonal estimate. At $\kappa\le40$
the circle method has no outputs. }% Source: \cite{esprit-comparison-data}.}
\label{tab:cmp-angle-errors}
\end{table}

For $\kappa\ge160$, the raw circle advantage in
\cref{tab:cmp-angle-errors} disappears after common refinement. At
$\kappa=160$, both methods recover all ten clouds within 12 MUSIC steps,
and a separate point-set comparison gives final agreement to about
$1.1\times10^{-16}$. At $\kappa=80$, the one
whole-sphere fixed-angle outlier still matches only 123 of 125 targets
after refinement, with bottleneck error about $0.05610$, even at
$K=125$. A small median therefore does not imply complete success.
In some low-frequency cases the MUSIC objective decreases while the
bottleneck error increases, as is possible when the objective and the
truth-based matching loss are different.

The 1\% additive cases show the same pattern. At $\kappa=80$, the
whole-sphere raw/final medians are $1.8032\times10^{-2}$ and
$1.0100\times10^{-5}$, with all ten clouds recovered. At
$\kappa=160,320,640,1280$, respectively, the raw whole-sphere/circle
medians are
\[
 (2.7290\times10^{-3},2.5920\times10^{-4}),\quad
 (7.8466\times10^{-4},7.9223\times10^{-6}),
\]
\[
 (2.7897\times10^{-4},1.9432\times10^{-6}),\quad
 (1.5793\times10^{-4},6.4037\times10^{-7}).
\]
Both methods recover all ten clouds at each of these frequencies; their
common final medians are $2.4115\times10^{-6}$,
$5.9282\times10^{-7}$, $1.4452\times10^{-7}$, and
$3.8539\times10^{-8}$.

%\subsubsection{A tested regime where paired circles do win after refinement}

The supplementary 5\% additive comparison uses the same noise
construction at $\kappa=160,320,640,1280$ on all ten clouds:
\begin{table}[htbp]
\centering\small
\begin{tabular}{rlrrc}
\toprule
$\kappa$&Initializer&Raw median $E_\infty$&After MUSIC&Complete\\
\midrule
160&Whole sphere&$1.1576\times10^{-1}$&$1.1646\times10^{-1}$&0/10\\
160&Paired circle&$1.2887\times10^{-3}$&$1.2133\times10^{-5}$&10/10\\
320&Whole sphere&$4.1272\times10^{-2}$&$4.3803\times10^{-2}$&1/10\\
320&Paired circle&$1.1538\times10^{-4}$&$2.9569\times10^{-6}$&10/10\\
640&Whole sphere&$3.3877\times10^{-2}$&$3.2866\times10^{-2}$&0/10\\
640&Paired circle&$1.8918\times10^{-5}$&$7.2257\times10^{-7}$&10/10\\
1280&Whole sphere&$2.0863\times10^{-2}$&$2.0845\times10^{-2}$&0/10\\
1280&Paired circle&$5.4692\times10^{-6}$&$1.9262\times10^{-7}$&10/10\\
\bottomrule
\end{tabular}
\caption{Baseline localization under 5\% additive snapshot noise,
normalized in the finite harmonic Frobenius norm before
orthonormalization. Whole sphere uses the original diagonal estimate.}
\label{tab:cmp-additive-five}
\end{table}
At 5\% additive perturbation, the circle advantage persists after the
common refinement, under the fixed protocol in \cref{tab:cmp-protocol}.

A 5\% equal-angle random-complement comparison at
$\kappa=1280$ holds every principal sine at $0.05$. The original
whole-sphere diagonal estimate has raw and refined medians $2.1089\times10^{-2}$ and
$2.0964\times10^{-2}$, with complete recovery in
$0/10$ clouds; the paired-circle medians are
$5.4871\times10^{-6}$ and $1.9333\times10^{-7}$, with
$10/10$ complete recoveries.
The difference therefore persists when the subspace angle is fixed. All 50
whole-sphere estimates at this wavenumber have numerical rank $375$ and
block singular-value ratios between $0.977$ and $0.982$: the incomplete
recoveries occur with well-conditioned coupled systems.
\Cref{sec:phase-readout-experiment} revisits these same generator eigenbases.

%\subsubsection{Noise labels and timing do not support broader claims}

At 1\% additive noise, the median realized subspace errors are approximately $0.995$, $0.0607$, $0.0176$, $0.0130$,
and $0.0114$ as $\kappa$ increases through the values $10,20,40,80,160$. Under the equal-angle model they are $0.01$ by
construction. At $\kappa=320,640$, the additive-noise medians are $0.01066$ and $0.01033$, with true synthesis singular-value ratios $0.900$ and $0.948$, respectively. At $\kappa=1280$, the corresponding subspace-error median and synthesis singular-value ratio are
$0.01017$ and $0.974$.
The true synthesis singular-value ratios at these five wavenumbers are approximately $1.50\times10^{-5}$, $0.0695$, $0.394$,
$0.652$, and $0.811$. Thus very small $\kappa$ is not automatically a
good practical regime for either method: signal-subspace extraction can
be ill conditioned, and at a fixed small subspace angle the algebraic
and spectral stages can still fail.

\subsection{A common harmonic input at two representative frequencies}
\label{sec:common-harmonic-benchmark}

We next ask whether the baseline comparison persists when both
initializers and MUSIC use the same finite harmonic data. We retain
the original diagonal generator coordinates and test a single cloud at
$\kappa=80$ and $1280$. The geometry, auxiliary centers, mixing matrix,
cutoff rule, candidate pencils, circle nodes, phase continuation, and
200-step refinement are unchanged. At each frequency we test noiseless
data and both 5\% perturbation models; at $\kappa=80$ we also include
the 1\% cases. The retained degrees are $(K,L)=(125,126)$ and
$(1236,1237)$, respectively.

For each case, form one orthonormal empirical coefficient frame
$\widehat Q_L=H_LC$ and its finite aperture expansion
$F_L(\bomega)=\sum_{\ell\le L,m}
Y_{\ell m}(\bomega)(\widehat Q_L)_{\ell m,:}$.
The generator equations use the guarded actions of this frame through
degree $K=L-1$. Every paired-circle sample is synthesized from the same
$\widehat Q_L$. Since $4000>2L$ at both frequencies, the uniform circle
grid integrates products of these retained restrictions exactly in exact
arithmetic. The common refinement objective is
\begin{equation}
 J_L(\by)=1-\norm{\widehat Q_L^*a_L(\by)}^2,
 \qquad
 a_L(\by)=\bigl(\langle Y_{\ell m},\varphi_{\by}\rangle
                  \bigr)_{\ell\le L,m}.
 \label{eq:common-harmonic-music}
\end{equation}
The leading $1$ is the known full-sphere norm squared of the atom, as in
\eqref{eq:cmp-music}; the projected correlation uses only retained
coefficients. For efficiency, these correlations and their derivatives
are evaluated from the factorization $H_LC$ by the finite addition
theorem, summing degrees $0,\ldots,L$, rather than the continuous sinc
kernel. 
\begin{table}[htbp]
\centering\small
\begin{tabular}{rllrrc}
\toprule
$\kappa$ & Perturbation & Initializer & Raw $E_\infty$ & After MUSIC & Coverage\\
\midrule
80 & None & Generator & $1.38\times10^{-15}$ & $1.14\times10^{-16}$ & $125\to 125$\\
 &  & Circles & --- & --- & No output\\
80 & 1\% additive & Generator & $1.34\times10^{-2}$ & $7.82\times10^{-6}$ & $123\to 125$\\
 &  & Circles & $5.85\times10^{-1}$ & $5.78\times10^{-1}$ & $0\to 2$\\
80 & 1\% equal-angle & Generator & $1.19\times10^{-2}$ & $8.18\times10^{-6}$ & $123\to 125$\\
 &  & Circles & $5.46\times10^{-1}$ & $5.41\times10^{-1}$ & $0\to 3$\\
80 & 5\% additive & Generator & $1.51\times10^{-1}$ & $1.41\times10^{-1}$ & $112\to 118$\\
 &  & Circles & $5.47\times10^{-1}$ & $5.47\times10^{-1}$ & $0\to 2$\\
80 & 5\% equal-angle & Generator & $1.64\times10^{-1}$ & $1.76\times10^{-1}$ & $111\to 118$\\
 &  & Circles & $5.42\times10^{-1}$ & $5.41\times10^{-1}$ & $0\to 4$\\
1280 & None & Generator & $8.29\times10^{-16}$ & $1.24\times10^{-16}$ & $125\to 125$\\
 &  & Circles & $3.39\times10^{-16}$ & $1.24\times10^{-16}$ & $125\to 125$\\
1280 & 5\% additive & Generator & $7.02\times10^{-2}$ & $7.02\times10^{-2}$ & $3\to 113$\\
 &  & Circles & $3.11\times10^{-6}$ & $2.21\times10^{-7}$ & $125\to 125$\\
1280 & 5\% equal-angle & Generator & $7.14\times10^{-2}$ & $7.14\times10^{-2}$ & $3\to 113$\\
 &  & Circles & $3.12\times10^{-6}$ & $2.21\times10^{-7}$ & $125\to 125$\\
\bottomrule
\end{tabular}
\caption{One common harmonic frame for both initializers and MUSIC. }
\label{tab:common-harmonic-input}
\end{table}

Table~\ref{tab:common-harmonic-input} gives the results. At $\kappa=80$
and 1\% noise, the whole-sphere initializations recover all targets after
refinement, whereas the circle initializations do not. At 5\% noise,
neither method recovers the entire cloud. At $\kappa=1280$ and 5\%
noise, the circle initializations recover all 125 targets under both
perturbation models, whereas the diagonal generator coordinates miss
targets. Thus the difference observed in the baseline comparison also
occurs with a common finite input; continuous circle evaluation and
continuous-kernel refinement do not alone account for it.

Restriction of $F_L$ does not, in general, preserve the exact
plane-wave shift identity. We therefore measure its defect using the
known continuous transfer matrix,
\begin{equation}
 d_L(q,\ba)=
 \frac{\norm{F_{L,+}-F_{L,-}T_{q,\ba}}_{\rm F}}
      {\norm{F_{L,+}}_{\rm F}},\qquad
 T_{q,\ba}=C_X^{-1}
       \diag\!\left(e^{\mathrm i q\ba\cdot\bx_j}\right)C_X,
 \label{eq:common-truncation-defect}
\end{equation}
where $F_{L,\pm}$ are sampled on the paired circles and $C_X$ is the
square signal coefficient factor in the noiseless frame
$\widehat Q_L=H_{X,L}C_X$. 

The maximum measured identity defects over all prescribed axes and
shifts are $8.42\times10^{-14}$ at $\kappa=80$ and $8.79\times10^{-13}$
at $\kappa=1280$. These measurements include harmonic-synthesis and
floating-point errors; they are not estimates of truncation error alone.
At $\kappa=80$, the first failed check is at $q=50$ on the second
coordinate axis: its restriction has numerical rank 124 at relative
threshold $10^{-12}$, with smallest-to-largest singular-value ratio
$3.74\times10^{-13}$. The continuous noiseless control fails at the
same step. At $\kappa=1280$, all prescribed noiseless restrictions pass
the rank check, and both methods recover the cloud to numerical
precision. Thus this cutoff yields no material noiseless localization
loss in the well-conditioned case. 
Perturbations can lift small restriction singular values and allow the
rank check to pass without yielding accurate circle coordinates, as the
noisy $\kappa=80$ results illustrate.

\subsection{Phase-gradient  readout at \texorpdfstring{$\kappa=1280$}{kappa = 1280}}
\label{sec:phase-readout-experiment}

For the continuous-frame control, we evaluate \(f_j(\boldsymbol\omega)=F(\boldsymbol\omega)\boldsymbol v_j\) and \(\nabla_{\mathbb S^2}f_j(\boldsymbol\omega)=(\nabla_{\mathbb S^2}F(\boldsymbol\omega))\boldsymbol v_j\) from the synthetic Fourier frame at the same quadrature nodes, using the same selected eigenvector \(\boldsymbol v_j\) as in the harmonic calculation.

We test \eqref{eq:phase-readout-main} on the same ten clouds under both
5\% noise models, holding the generator matrices, selected common-pencil
eigenvectors, and MUSIC iteration fixed.  The finite-data
implementation forms \eqref{eq:phase-harmonic-reconstruction}
through the full input degree $L=K+1$ and synthesizes both angular
derivatives from those coefficients. Across the ten clouds,
$L\in\{1233,1234,1235,1237\}$ selected by the rule \eqref{eq:cmp-cutoff-rule}. 

The main quadrature uses $N=262144$ equal-area Fibonacci sphere nodes,
rotated using random seed 1729. With
\[
 \ba_{j,n}=\frac1\kappa\operatorname{Im}
       \left(\frac{\nabla_{\Sph^2}f_j(\bomega_n)}
                        {f_j(\bomega_n)}\right),
 \qquad P_n=I_3-\bomega_n\bomega_n^T,
\]
we compute
\begin{equation}
 \bx_{j,N}^{\,\rm ph}
 =\left(\sum_{n=1}^N P_n\right)^{-1}
       \sum_{n=1}^N\ba_{j,n}.
 \label{eq:phase-readout-discrete}
\end{equation}
All estimates receive the same 200 MUSIC steps and
are scored at the same tolerance $\rho=7/(12\kappa)$.

\begin{table}[htbp]
\centering\small
\begin{tabular}{llrrcc}
\toprule
Noise & Generator estimate & Raw median $E_\infty$ & After MUSIC
& Raw complete & Final complete\\
\midrule
Additive & Diagonal & $2.08626\times10^{-2}$ & $2.08449\times10^{-2}$ &0/10&0/10\\
Additive & Phase, continuous & $2.28345\times10^{-4}$ & $1.92618\times10^{-7}$ &9/10&10/10\\
Additive & Phase, harmonic & $2.28345\times10^{-4}$ & $1.92618\times10^{-7}$ &9/10&10/10\\
Equal angle & Diagonal & $2.10891\times10^{-2}$ & $2.09643\times10^{-2}$ &0/10&0/10\\
Equal angle & Phase, continuous & $2.30431\times10^{-4}$ & $1.93334\times10^{-7}$ &9/10&10/10\\
Equal angle & Phase, harmonic & $2.30431\times10^{-4}$ & $1.93334\times10^{-7}$ &9/10&10/10\\
\bottomrule
\end{tabular}
\caption{Whole-sphere generator coordinates at $\kappa=1280$ and 5\%
perturbation. The harmonic rows reconstruct the demixed functions and
their angular derivatives solely from the retained coefficients.}
\label{tab:phase-normalized-1280}
\end{table}

In \cref{tab:phase-normalized-1280}, phase-gradient readout increases complete
recovery after refinement from $0/10$ to $10/10$ under each noise model,
without changing the generator eigenbasis. The harmonic and
continuous-frame coordinates differ pointwise by at most
$4.92\times10^{-15}$ over both quadrature grids and all twenty cases;
their final coordinates agree exactly in the saved arithmetic. The
twenty refined point sets also agree with the paired-circle results to
within $5.56\times10^{-17}$ in bottleneck distance. On these fixed grids,
the retained harmonic coefficients therefore suffice to evaluate the
phase-gradient readout to the reported accuracy, without analytic differentiation
of the underlying synthetic frame.

\commentout{
Independent checks on three demixed functions per case, each at 256 new
nodes, give maximum discrepancies $3.72\times10^{-12}$ in function
value and $3.46\times10^{-9}$ in angular derivative, or
$2.71\times10^{-12}$ after division by $\kappa$. A resolved one-atom
test recovers its coordinate to $1.50\times10^{-15}$, while a deliberately
underresolved test reproduces the truncated harmonic function rather
than the continuous plane wave. The minimum sampled demixed amplitude
on the main grid is $0.38075$; this is a diagnostic, not a global
nonvanishing certificate.
}

Quadrature remains a separate limitation. Repeating all twenty cases at
$N=65536$ changes a raw coordinate by as much as
$8.99\times10^{-4}$. At $\kappa=1280$ the scoring tolerance is
\[
 \rho=\frac{7}{12\kappa}=4.56\times10^{-4},
\]
so the observed grid-to-grid displacement is about twice the localization
tolerance. \commentout{Grid convergence of the raw phase integral has therefore not
been demonstrated at the scale used to declare recovery. The refined point sets nevertheless
change by at most $1.11\times10^{-16}$, showing only that both grids put
one proposal in each of the same empirical MUSIC attraction basins. Thus
the experiment supports basin coverage after refinement, not
high-accuracy convergence of the unrefined phase estimator. It does not
repair the original diagonal generator coordinates, whose median errors
are about two orders of magnitude larger and whose target coverage is
incomplete.
}

A post-hoc examination of the demixed coefficients shows that the
largest coefficient in each of the 125 columns corresponds to a distinct
target in every case. This observation distinguishes the quality of the
common eigenbasis from that of its diagonal coordinates: the same basis
gives successful initializations when the coordinates are read from the
normalized phases.

\section{Conclusion and prospects}
\label{sec:outlook}
The two spherical ESPRIT constructions recover the same commuting
coordinate matrices from different aperture identities. Paired
small-circle restrictions determine their exponentials, whereas the
whole-sphere generator equations determine the matrices directly.
The Hodge formulation and the synchronized great-circle limit relate
these exact identities. The two constructions thus offer complementary ways
of organizing the geometric information in a spherical Fourier subspace. 

The phase-gradient method separates the use of a generator eigenbasis for
demixing from the use of its diagonal entries for localization.
Under pointwise dominance, the phase-gradient estimate has an explicit
$O(\kappa^{-1})$ error bound; for a two-atom mixture with fixed
nonzero separation, the sharper bound is $O(\kappa^{-2})$.
The numerical examples show that retained harmonic coefficients can
produce useful phase initializations even when the original diagonal
coordinates fail. This suggests
using the whole-sphere generator equations for demixing and then recovering locations from the
phase variation of the demixed functions.  

The local form of the generator identities also permits a formulation
on an incomplete aperture. For forward-hemisphere observations, take
the incident direction as the north axis and set
$\Gamma=\{\bomega\in\Sph^2:\omega_3>0\}$. The generator equations
remain valid on $\Gamma$, using the observed frame and its local
angular derivatives. In the exact finite-atom model, a null field on
this open aperture extends analytically to a null field on the whole
sphere; \cref{thm:qualitative-injectivity} then gives injectivity.
This argument is qualitative and does not give an aperture-dependent
conditioning bound. Alternatively,
\cref{rem:hemispherical-paired-arcs} gives a derivative-free paired-arc
construction on the forward hemisphere, with a tradeoff between arc
coverage and phase sensitivity.
\appendix

\section{Aperture algebra and commutators}
\label{app:aperture-algebra}

We collect the aperture commutators here to distinguish them from
the commuting coefficient matrices in the two recovery constructions.
All identities hold on the common smooth domain
$C^\infty(\Sph^2)$.

\begin{proposition}
\label{prop:se3}
Define the aperture-side coordinate multiplication 
$$
  M_rf(\bomega)=\omega_rf(\bomega),\quad r=1,2,3,
 $$
 on \(C^\infty(\Sph^2)\). 
Then 
\begin{align}
  [\mathcal L_i,\mathcal L_j]
  &=i\sum_{k=1}^3\epsilon_{ijk}\mathcal L_k,
  \label{eq:LL-commutator}\\
  [\mathcal L_i,M_q]
  &=i\sum_{r=1}^3\epsilon_{iqr}M_r,
  \label{eq:LM-commutator}\\
  [M_q,M_r]&=0.
  \label{eq:MM-commutator}
\end{align}
Thus \((\boldsymbol{\mathcal L},\bm M)\) realizes the semidirect
product \(\mathfrak{so}(3)\ltimes\R^3=\mathfrak{se}(3)\) on the aperture.
\end{proposition}

\begin{proof}
Put $D_i=(\be_i\times\bomega)\cdot\nabla_{\Sph^2}$, so
$\mathcal L_i=-iD_i$. For the linear vector fields
$K_i(\bomega)=[\be_i]_\times\bomega$, the vector-field bracket is
\[
 [K_i,K_j](\bomega)
 =\bigl([\be_j]_\times[\be_i]_\times
       -[\be_i]_\times[\be_j]_\times\bigr)\bomega
 =-\sum_k\epsilon_{ijk}K_k(\bomega).
\]
Consequently $[D_i,D_j]=-\sum_k\epsilon_{ijk}D_k$, and
$[\mathcal L_i,\mathcal L_j]
=- [D_i,D_j]=i\sum_k\epsilon_{ijk}\mathcal L_k$.
For the mixed relation,
\[
  [\mathcal L_i,M_q]f
  =(\mathcal L_i\omega_q)f
  =i\sum_r\epsilon_{iqr}\omega_rf.
\]
The multiplication operators commute, proving the last relation.
\end{proof}

The commuting half of this algebra exponentiates to aperture modulations
\begin{equation}
\label{eq:aperture-translations}
  T_{\bh}=\exp\!\left(i\kappa\sum_{r=1}^3h_rM_r\right),
  \qquad
  (T_{\bh}f)(\bomega)=e^{i\kappa\bomega\cdot\bh}f(\bomega),
\end{equation}
and \(T_{\bh}\varphi_{\bx}=\varphi_{\bx+\bh}\).  The rotation and
translation parts therefore have different geometric roles even though they
belong to one Euclidean Lie algebra.

\section{Cocycle intertwining and algebraic compatibility}
%\subsection{Lie-algebra compatibility}
\label{sec:compatibility}

Let
\(
  \mathsf R_{\ba}(\theta)=\exp(\theta[\ba]_\times)
\)
for \(\ba\in\Sph^2\), and define
\begin{equation}
\label{eq:finite-rotation-cocycle}
  \boldsymbol\delta_{\ba}(\theta,\bomega)
  :=\mathsf R_{\ba}(\theta)\bomega-\bomega,
  \qquad
  \mathcal C_{\ba}(\theta,\bomega)
  :=\Phi\bigl(\boldsymbol\delta_{\ba}(\theta,\bomega)\bigr).
\end{equation}

\begin{proposition}
\label{prop:rotation-cocycle}
The finite rotation obeys the pointwise intertwining relation
\begin{equation}
\label{eq:rotation-pointwise}
  Q(R_{\ba}(\theta)\bomega)
  =Q(\bomega)\cC_{\ba}(\theta,\bomega).
\end{equation}
The multiplier is a transformation-group cocycle:
\begin{equation}
\label{eq:cocycle-law}
  \cC_{\ba}(\theta+\sigma,\bomega)
  =\cC_{\ba}(\theta,\bomega)
   \cC_{\ba}(\sigma,R_{\ba}(\theta)\bomega).
\end{equation}
Its first two derivatives at the identity are
\begin{align}
  \left.\frac1{i\kappa}\partial_\theta
  \cC_{\ba}(\theta,\bomega)\right|_{\theta=0}
  &=\Psi_{\ba\times\bomega},
  \label{eq:cocycle-first}\\
  \left.\partial_\theta^2
  \cC_{\ba}(\theta,\bomega)\right|_{\theta=0}
  &=i\kappa\Psi_{\ba\times(\ba\times\bomega)}
   -\kappa^2(\Psi_{\ba\times\bomega})^2.
  \label{eq:cocycle-second}
\end{align}
\end{proposition}

\begin{proof}
Equation \eqref{eq:rotation-pointwise} is \eqref{eq:any-two-points} with
\(\bomega'=R_{\ba}(\theta)\bomega\).  The displacement identity
\[
  \boldsymbol\delta_{\ba}(\theta+\sigma,\bomega)
  =\boldsymbol\delta_{\ba}(\theta,\bomega)
   +\boldsymbol\delta_{\ba}
      (\sigma,R_{\ba}(\theta)\bomega)
\]
together with \eqref{eq:Phi-group-law} proves the cocycle law.  Finally,
\[
  \boldsymbol\delta_{\ba}'(0,\bomega)=\ba\times\bomega,
  \qquad
  \boldsymbol\delta_{\ba}''(0,\bomega)
  =\ba\times(\ba\times\bomega),
\]
and differentiation of \(e^{i\kappa\Psi_{\boldsymbol\delta}}\) proves
\eqref{eq:cocycle-first}--\eqref{eq:cocycle-second}.  All matrices in the
exponent commute because they are linear combinations of the \(\Psi_r\).
\end{proof}

For a fixed axis \(\ba\), define the matrix field
\begin{equation}
\label{eq:F-a}
  F_{\ba}(\bomega):=\Psi_{\ba\times\bomega}.
\end{equation}
Then \eqref{eq:fixed-axis-variable-pencil} reads
\begin{equation}
\label{eq:LQ-QF}
  \mathcal L_{\ba}Q=\kappa QF_{\ba}.
\end{equation}
The fields \(F_{\ba}\) vary with the base point even though the underlying
Cartesian matrices \(\Psi_r\) are constant.

\begin{proposition}[Maurer--Cartan compatibility]
\label{prop:Maurer-Cartan}
For \(\ba,\bb\in\R^3\),
\begin{equation}
\label{eq:F-derivative-identity}
  \mathcal L_{\ba}F_{\bb}
  -\mathcal L_{\bb}F_{\ba}
  =iF_{\ba\times\bb}.
\end{equation}
For the ESPRIT fields \eqref{eq:F-a}, the vanishing commutators
\([\Psi_r,\Psi_t]=0\) imply the Maurer-Cartan compatibility relation
\begin{equation}
\label{eq:Maurer-Cartan}
  \mathcal L_{\ba}F_{\bb}
  -\mathcal L_{\bb}F_{\ba}
  -iF_{\ba\times\bb}
  +\kappa[F_{\ba},F_{\bb}]=0.
\end{equation}

\end{proposition}

\begin{proof}
Because $\Psi_{\mathbf z}$ is linear in $\mathbf z$, direct
differentiation gives
\[
 \mathcal L_\ba F_\bb(\bomega)
 =-i\Psi_{\bb\times(\ba\times\bomega)}.
\]
Hence
\begin{align*}
 \mathcal L_\ba F_\bb-\mathcal L_\bb F_\ba
 &=-i\Psi_{\bb\times(\ba\times\bomega)
                 -\ba\times(\bb\times\bomega)}\\
 &=i\Psi_{(\ba\times\bb)\times\bomega}
 =iF_{\ba\times\bb},
\end{align*}
where the second equality is the vector triple-product identity. Since
each $F_\ba(\bomega)$ is a linear combination of the commuting matrices
$\Psi_1,\Psi_2,\Psi_3$, one has $[F_\ba,F_\bb]=0$. Substitution into
\eqref{eq:Maurer-Cartan} completes the proof.
\end{proof}
In noisy computations, both algebraic structures become useful diagnostics:
\(
[\widehat\Psi_r,\widehat\Psi_t]
\)
tests the translation representation, while a discrete residual in
\eqref{eq:Maurer-Cartan} tests rotational compatibility.

\section{Proof of the generic harmonic cutoff}
\label{app:generic-cutoff}

\begin{proof}[Proof of \cref{thm:rank-generic-cutoff}]
We first exhibit one configuration for which a fixed collection of
harmonic tests gives a nonsingular matrix.  Choose collinear nodes
\[
 \bx_j=\frac{t_j}{\kappa}\be_3,
\]
and, for \(\ell\ge1\), introduce the real degree-\(\ell\) harmonics
\begin{equation}
\label{eq:rank-test-harmonics}
 h_\ell^{(1)}(\bomega)=\omega_1P_\ell'(\omega_3),
 \qquad
 h_\ell^{(2)}(\bomega)=\omega_2P_\ell'(\omega_3).
\end{equation}
They are proportional to the real and imaginary harmonics of angular
order one.  With the inner product conjugate-linear in its first argument
and normalized surface measure, azimuthal integration gives
\begin{equation}
\label{eq:rank-test-integral}
 \left\langle h_\ell^{(a)},\omega_b e^{it\omega_3}\right\rangle
 =\delta_{ab}\,d_\ell\frac{j_\ell(t)}t,
 \qquad
 d_\ell=\frac{\ell(\ell+1)}2\,i^{\ell-1},
 \quad a,b\in\{1,2\}.
\end{equation}
For completeness, the diagonal integral equals
\[
 \frac14\int_{-1}^1(1-u^2)P_\ell'(u)e^{itu}\,du.
\]
The Legendre equation and integration by parts yield, for \(t\ne0\),
\[
 \int_{-1}^1(1-u^2)P_\ell'(u)e^{itu}\,du
 =\frac{\ell(\ell+1)}{it}\int_{-1}^1P_\ell(u)e^{itu}\,du
 =2\ell(\ell+1)i^{\ell-1}\frac{j_\ell(t)}t.
\]
Here the last identity follows from the plane-wave expansion; see
\href{https://dlmf.nist.gov/10.60.E7}{DLMF, equation 10.60.7}.
The formulas extend continuously to \(t=0\).

Define
\begin{equation}
\label{eq:rank-bessel-matrix}
 T_s(t_1,\ldots,t_s)
 =\left[\frac{j_\ell(t_j)}{t_j}\right]_{\ell,j=1}^s.
\end{equation}
At zero, the quotient has value \(1/3\) for \(\ell=1\) and zero for
\(\ell\ge2\).  Since \(K\ge s\), all test functions
\eqref{eq:rank-test-harmonics} with \(1\le\ell\le s\) belong to
\(\cH_K\).  Taking inner products with these test functions is therefore
unaffected by the retained orthogonal projection.  Testing the concatenated raw multiplier map
against these \(2s\) functions produces the square matrix
\begin{equation}
\label{eq:rank-tested-pair}
 \begin{bmatrix}D_sT_s&0\\0&D_sT_s\end{bmatrix},
 \qquad D_s=\diag(d_1,\ldots,d_s).
\end{equation}
All \(d_\ell\) are nonzero.  Hence nonsingularity of \(T_s\)
establishes \eqref{eq:rank-generic-pair} for this collinear configuration.

Choose distinct real numbers \(a_1,\ldots,a_s\), and put
\(t_j=\epsilon a_j\).  The spherical-Bessel power series gives
\[
 \frac{j_\ell(t)}t
 =\frac{t^{\ell-1}}{(2\ell+1)!!}\bigl(1+O(t^2)\bigr),
 \qquad \ell\ge1;
\]
see \href{https://dlmf.nist.gov/10.53.E1}{DLMF, equation 10.53.1}.
Factoring \(\epsilon^{\ell-1}\) from row \(\ell\) gives
\begin{equation}
\label{eq:rank-vandermonde-asymptotics}
 \det T_s(\epsilon a_1,\ldots,\epsilon a_s)
 =\epsilon^{s(s-1)/2}
 \left[
 \frac{\displaystyle\prod_{p<q}(a_q-a_p)}
      {\displaystyle\prod_{\ell=1}^s(2\ell+1)!!}
 +O(\epsilon^2)\right].
\end{equation}
The leading coefficient is nonzero, including when one of the \(a_j\)
is zero.  Thus \(T_s\) is nonsingular for all sufficiently small
positive \(\epsilon\).  Taking also
\(\epsilon\max_j|a_j|<\kappa\) puts every node in this configuration inside
\(B_1(0)\).

For arbitrary node positions, use the same fixed \(2s\) harmonic tests
to form a square matrix from
\([N_{1,K}(X)\ N_{2,K}(X)]\).  Its entries are integrals of fixed
polynomials times \(e^{i\kappa\bomega\cdot\bx_j}\), hence are real
analytic in the \(3s\) real node coordinates.  Let \(d(X)\) be its
determinant.  The collinear construction proves that \(|d(X)|^2\) is a
nontrivial real-analytic function on \(\R^{3s}\).  Its zero set has
empty interior and Lebesgue measure zero.  The set
\(\{X\in\Omega_s:d(X)\ne0\}\) is therefore dense and open in
\(\Omega_s\), with a measure-zero complement, and satisfies
\eqref{eq:rank-generic-pair}.

The argument following \eqref{eq:rank-basis-factorization} now gives
full rank of \(\Phi_{X,L}\) for every \(L\ge K+1\).
Equation \eqref{eq:rank-basis-factorization} and
\cref{prop:rank-two-multipliers} prove the asserted guarded-block rank.
\end{proof}

\section{One-atom bias }
\label{app:scalar-kernel}

We compare the two estimates when a one-target signal subspace is
perturbed by an atom at another location. The observed subspace is
one-dimensional, so neither eigenbasis selection nor coordinate pairing
is involved. The generator method fits an intensity-weighted
differential relation, whereas the paired-circle method estimates a
finite-shift phase. The resulting formulas identify the bias in the
first estimate and the conditions under which the second gives an error
of order $\kappa^{-1}$.

Let $\kappa>0$, let $\bx\ne\bz$ be points in the unit ball, and define
\[
 \varphi_{\bx}(\bomega)=e^{i\kappa\bomega\cdot\bx},\qquad
 \bomega\in\Sph^2,\qquad d\mu(\bomega)=\frac{d\bomega}{4\pi}.
\]
All spherical inner products below use $d\mu$ and are conjugate-linear
in the first argument. In particular, every atom has unit norm. The true
signal subspace is $\operatorname{span}\{\varphi_{\bx}\}$, while the
observed subspace is generated by
\begin{equation}
 f=A\varphi_{\bx}+B\varphi_{\bz},\qquad A,B\in\C,\qquad f\ne0
 \label{eq:cmp-one-mixture}
\end{equation}
where $\bx$ is the target and the term at $\bz$ represents a perturbing
component. 

Put $\bd=\bz-\bx$, $r=|\bd|$, and $\bu=\bd/r$. The atom correlation  is
\[
 a_0=\langle\varphi_{\bx},\varphi_{\bz}\rangle
     =\sinc(\kappa r),\qquad
 \sinc(t)=\frac{\sin t}{t},\quad \sinc(0)=1.
\]
To prescribe the subspace perturbation exactly, choose $0\le\eta<1$
and set $c=\sqrt{1-\eta^2}$. Since $|a_0|<1$, define
\begin{equation}
 \begin{aligned}
 \psi&=\frac{\varphi_{\bz}-a_0\varphi_{\bx}}{\sqrt{1-a_0^2}},
 & f&=c\varphi_{\bx}+\eta \psi,\\
 A&=c-\frac{\eta a_0}{\sqrt{1-a_0^2}},
 & B&=\frac{\eta}{\sqrt{1-a_0^2}}.
 \end{aligned}
 \label{eq:cmp-one-equal-angle}
\end{equation}
Then $\psi\perp\varphi_{\bx}$, $\|\psi\|=\|f\|=1$, and the angle between
the true and observed subspaces has sine $\eta$. 

\subsection{Whole-sphere least-squares estimate}

Write $P_{\bomega}=I_3-\bomega\bomega^T$ for the tangent-plane projector.
For an unperturbed atom,
\[
 \frac{1}{i\kappa}\nabla_{\Sph^2}\varphi_{\bx}
       =P_{\bomega}\bx\,\varphi_{\bx}.
\]
The generator construction seeks a single constant vector $\by$ that
best satisfies the same relation for the observed function $f$:
\begin{equation}
 \by_g=\operatorname*{argmin}_{\by\in\R^3}
 \int_{\Sph^2}
 \left|\frac{1}{i\kappa}\nabla_{\Sph^2}f(\bomega)
              -P_{\bomega}\by\,f(\bomega)\right|^2\,d\mu.
 \label{eq:cmp-generator-objective}
\end{equation}
Applying $\bomega\times$ to this tangential residual preserves its norm,
so this is also the angular-generator least-squares problem. For the
mixture in \eqref{eq:cmp-one-mixture}, its residual is
\[
 P_{\bomega}\bigl[A(\bx-\by)\varphi_{\bx}
                     +B(\bz-\by)\varphi_{\bz}\bigr].
\]
Thus the derivative relation contains contributions from both locations,
although the requested output is one vector.

The interaction of these contributions is described by the real,
symmetric matrix kernel $T_\kappa$ in
\eqref{eq:tangential-kernel}. Put
\[
 c_{AB}=\operatorname{Re}(\overline A B).
\]
The averaging identity
$\int P_{\bomega}\,d\mu=2I_3/3$ gives the following exact expression.

\begin{proposition}
\label{prop:cmp-scalar-generator}
The minimizer in \eqref{eq:cmp-generator-objective} is unique and satisfies
\begin{equation}
  \by_g=H^{-1}\left[
       \frac23(|A|^2\bx+|B|^2\bz)
                  +c_{AB}T_\kappa(\bd)(\bx+\bz)\right]
 \label{eq:cmp-exact-yG}
\end{equation}
where 
$$H=\frac23(|A|^2+|B|^2)I_3+2c_{AB}T_\kappa(\bd)
$$ is positive definite. 

\end{proposition}
\begin{proof}
Expanding the residual gives the normal matrix
$H=\int|f|^2P_{\bomega}\,d\mu$ and the right-hand vector
\[
 \operatorname{Re}\int\overline f\,
 P_{\bomega}(A\bx\varphi_{\bx}+B\bz\varphi_{\bz})\,d\mu.
\]
The two pure-atom terms and the two cross terms give
\eqref{eq:cmp-exact-yG}. For any nonzero real vector $\ba$,
$\ba^TH\ba=\int |f|^2|P_{\bomega}\ba|^2\,d\mu>0$:
$f$ is a nonzero analytic function, and $P_{\bomega}\ba$ vanishes
only at the two poles parallel to $\ba$. Hence $H$ is positive definite.
\end{proof}

When the weighted cross-correlation contributions to both the normal
matrix and the right-hand vector are negligible,
\eqref{eq:cmp-exact-yG} reduces to
\begin{equation}
 \by_g\simeq\frac{|A|^2\bx+|B|^2\bz}{|A|^2+|B|^2}.
 \label{eq:cmp-generator-mixture}
\end{equation}
The estimate is approximately an energy-weighted average of the two
locations. Decay of the cross correlation does not remove the term
$|B|^2\bz$. Consequently, increasing the wavenumber need not remove the
localization bias.

The kernel can equivalently be computed by
\[
 T_\kappa(\bd)=\sinc(\kappa r)I_3
               +\kappa^{-2}\nabla_{\bd}^2\sinc(\kappa r).
\]
For $t=\kappa r>0$, the radial Hessian and
$\sinc''(t)+2\sinc'(t)/t+\sinc(t)=0$ give
\eqref{eq:tangential-kernel-formula}.
The expansion $\sinc(t)=1-t^2/6+O(t^4)$ gives
$T_\kappa(0)=2I_3/3$. For fixed $\bd\ne0$, both coefficients in
\eqref{eq:tangential-kernel-formula} tend to zero as $\kappa\to\infty$.

The formula also makes the finite-frequency bias explicit. Define
$\tau_\parallel(t)=-2\sinc'(t)/t$, the eigenvalue of $T_\kappa(\bd)$
in the direction $\bd$. Subtracting $H\bx$ from the normal equations gives
\begin{equation}
 \by_g-\bx=
 \frac{\frac23|B|^2+c_{AB}\tau_\parallel(t)}
      {\frac23(|A|^2+|B|^2)+2c_{AB}\tau_\parallel(t)}\,(\bz-\bx).
 \label{eq:cmp-generator-displacement}
\end{equation}
Thus, also for complex coefficients, the estimate lies on the line through the
two locations, although cross terms can prevent it from being a convex
average at finite frequency.

In the exact equal-angle construction, keep $\bx$, $\bz$, and $\eta$
fixed while increasing $\kappa$. Then $a_0\to0$, $A\to c$, and
$B\to\eta$, so
\begin{equation}
 \by_g-\bx\longrightarrow\eta^2(\bz-\bx).
 \label{eq:cmp-generator-floor}
\end{equation}
Thus a fixed subspace perturbation produces a nonzero limiting bias
even though $H\to2I_3/3$. The bias is already present in the continuous
one-dimensional problem; it does not arise from harmonic truncation,
numerical differentiation, or an ill-conditioned normal matrix.

\subsection{Paired-small-circle estimate}

Choose a unit vector $\bv$ and a chord length $0<h<2$, independently
of the perturbation size $\eta$. Set
\[
 q=\kappa h,\qquad
 E_{\bv}=\{\bomega\in\Sph^2:\bomega\cdot\bv=0\},\qquad
 \bomega_\pm=\sqrt{1-h^2/4}\,\bomega\pm\frac h2\bv.
\]
The paired points lie on the sphere and differ by $h\bv$. Define
$f_\pm(\bomega)=f(\bomega_\pm)$ on the common parameter circle
$E_{\bv}$, equipped with normalized arclength. An unperturbed atom obeys
\[
 (\varphi_{\bx})_+=e^{iqt_{\bx}}(\varphi_{\bx})_-,
 \qquad t_{\bx}=\bv\cdot\bx.
\]
For the observed function, the paired-circle construction estimates this
phase multiplier by
\begin{equation}
 \lambda_{h\bv}
 =\operatorname*{argmin}_{\lambda\in\C}\|f_+-\lambda f_-\|^2
 =\frac{\langle f_-,f_+\rangle}{\|f_-\|^2},
 \qquad \|f_-\|>0.
 \label{eq:cmp-circle-estimator}
\end{equation}
Only its phase is used to recover the projected location. 

Put
\[
 \delta_{\bv}=\bv\cdot(\bz-\bx),\qquad
 b_{\bv}=|(I_3-\bv\bv^T)(\bz-\bx)|,\qquad
 b_h=J_0\!\left(\kappa\sqrt{1-h^2/4}\,b_{\bv}\right),
\]
where $J_0$ is the cylindrical Bessel function. Here $b_{\bv}$ is the
separation after projection onto the circle plane, while $b_h$ is the
resulting circular correlation. They are different from the spherical
overlap $a_0=\sinc(\kappa r)$.

\begin{proposition}[Paired-circle phase and projection error]
\label{prop:cmp-scalar-circle}
The estimate in \eqref{eq:cmp-circle-estimator} is
\begin{equation}
 \lambda_{h\bv}=e^{iqt_{\bx}}
 \frac{|A|^2+|B|^2e^{iq\delta_{\bv}}
       +2c_{AB}b_h e^{iq\delta_{\bv}/2}}
      {|A|^2+|B|^2
       +2b_h\operatorname{Re}(\overline A B e^{-iq\delta_{\bv}/2})}.
 \label{eq:cmp-exact-lambda}
\end{equation}
If $A\ne0$ and
 \begin{equation}
 \beta_h=\frac{|B|^2+2|A|\,|B|\,|b_h|}{|A|^2}<1,
 \label{eq:cmp-scalar-beta}
\end{equation}
then the wrapped phase error obeys
$|\arg(\lambda_{h\bv}e^{-iqt_{\bx}})|
\le\arcsin\beta_h$. Once the correct phase branch has been chosen,
\begin{equation}
 |\widehat t_{\bx}-t_{\bx}|
 \le\frac{\arcsin\beta_h}{q}
 =\frac{\arcsin\beta_h}{\kappa h}.
 \label{eq:cmp-scalar-phase-bound}
\end{equation}
\end{proposition}
\begin{proof}
Expanding $\langle f_-,f_+\rangle$ gives the two pure-atom contributions
 $|A|^2e^{iqt_{\bx}}$ and
$|B|^2e^{iq(t_{\bx}+\delta_{\bv})}$. The two averaged cross terms sum to
$2c_{AB}b_h e^{iq(t_{\bx}+\delta_{\bv}/2)}$. The denominator in
\eqref{eq:cmp-exact-lambda} is $\|f_-\|^2>0$.
After removing $e^{iqt_{\bx}}$ and dividing the numerator by $|A|^2$,
the remaining complex number lies in the disk centered at $1$ with
radius $\beta_h$. If $\beta_h<1$, this disk avoids zero and subtends
an angle of at most $\arcsin\beta_h$. The positive real denominator
changes no phase. Dividing the correctly unwrapped phase error by $q$
proves \eqref{eq:cmp-scalar-phase-bound}.
\end{proof}

\subsection{A long shift reduces coordinate error}
\label{sec:long-shift}
A larger shift converts a phase error into a smaller coordinate error once the correct phase branch is known. 

For fixed $h\in(0,2)$ and fixed projected separation $b_{\bv}>0$,
$b_h\to0$ as $\kappa\to\infty$. In the equal-angle model, neglecting
the correlations gives
\begin{equation}
 \begin{aligned}
 \lambda_{h\bv}&\simeq e^{iqt_{\bx}}
             \bigl[(1-\eta^2)+\eta^2e^{iq\delta_{\bv}}\bigr],\\
 |\widehat t_{\bx}-t_{\bx}|&\lesssim
       \frac{1}{\kappa h}
       \arcsin\!\left(\frac{\eta^2}{1-\eta^2}\right),
       \qquad 0\le\eta<\frac1{\sqrt2}.
 \end{aligned}
 \label{eq:cmp-circle-approx}
\end{equation}
The exact bound has $\beta_h\to\eta^2/(1-\eta^2)$, so it gives
$O(\kappa^{-1})$ projection error for fixed admissible $\eta$ and $h$,
on the correct branch. Its limiting small-noise coefficient is of order
$\eta^2/h$. The gain comes from dividing a controlled phase error by a
long Fourier displacement $q=\kappa h$.

This conclusion has three qualifications. First, at finite frequency the
cross term of size $O(\eta|b_h|)$ can exceed the $O(\eta^2)$ contribution.
The latter dominates only when circular correlation is sufficiently small
relative to the perturbation; spherical overlap also affects $A$ and $B$.
Second, projected separation matters: if $b_{\bv}=0$, then $b_h=1$ and
circular decorrelation never occurs. Allowing $h$ to approach $2$ can
likewise shrink the circles too quickly for decorrelation. Third, a long
shift supplies the projection only modulo $2\pi/q$. A phase bound alone
does not select the correct integer.

A baseline shift $0<q_0=\kappa \eps_0<\min\{\pi,2\kappa\}$ is  unaliased for an exact atom in the unit ball. With a phase
error bounded by $e_0$, the margin $q_0|t_{\bx}|+e_0<\pi$ also prevents
an initial branch crossing. Continuation to larger shifts uses the
previous projection estimate to choose the next integer. If the wrapped phase error at shift $q_m$ is
bounded by $e_m$, a sufficient condition for that choice to be correct is
\[
 q_m|\widehat t_{m-1}-t_{\bx}|+e_m<\pi.
\]
Thus longer shifts require reliable continuation. In contrast, expanding
the decorrelated expression in \eqref{eq:cmp-circle-approx} for small $q$ gives
$\arg[(1-\eta^2)+\eta^2e^{iq\delta_{\bv}}]/q
\to\eta^2\delta_{\bv}$. In this approximation, small shifts retain the same weighted-average
bias. The improvement requires a long shift with the correct phase branch.

For the observed subspace $\operatorname{span}\{f\}$, the normalized
MUSIC objective is
\begin{equation}
 \begin{aligned}
 J_f(\by)=1-
 \frac{\left|A\sinc(\kappa|\by-\bx|)
              +B\sinc(\kappa|\by-\bz|)\right|^2}
 {|A|^2+|B|^2+2a_0\operatorname{Re}(\overline A B)}
 \end{aligned}
 \label{eq:cmp-one-music}
\end{equation}

The implications of these coordinate errors for MUSIC's $1/\kappa$
spatial scale and target coverage are discussed in
\cref{sec:music-initialization}.

\section{Phase-based estimate with amplitude normalization}
\label{app:phase-normalization}

The bias in \eqref{eq:cmp-generator-floor} is a property of the original
least-squares coordinate estimate. To see how its weighting enters,
write $Df=(i\kappa)^{-1}\nabla_{\Sph^2}f$. Wherever $f\ne0$, its objective
has integrand
\[
 |Df-P_{\bomega}\by f|^2
 =|f|^2\left|\frac{Df}{f}-P_{\bomega}\by\right|^2.
\]
Thus the observed intensity weights the local logarithmic derivative.
Removing this weight changes the coordinate normal equations. For real
$\by$, the amplitude-gradient component of $Df/f$ is purely imaginary
and contributes only a term independent of $\by$; taking the real phase
field, equivalently using $g=f/|f|$, removes that contribution.

\begin{proposition}
\label{prop:phase-readout}
Let $f\in C^1(\Sph^2;\C)$ be nonvanishing, and put $g=f/|f|$.
The unique real-coordinate minimizer
\begin{equation}
\bx^{\rm ph}
 =\operatorname*{argmin}_{\by\in\R^3}
 \int_{\Sph^2}\left|
       \frac{1}{i\kappa}\nabla_{\Sph^2}g
       -P_{\bomega}\by\,g\right|^2d\mu
 \label{eq:phase-normalized-objective}
\end{equation}
is
\begin{equation}
 \bx^{\rm ph}
 =\frac{3}{2\kappa}\int_{\Sph^2}
       \operatorname{Im}\left(\frac{\nabla_{\Sph^2}f}{f}\right)d\mu
 \label{eq:phase-normalized-estimator}
\end{equation}
where $d\mu$ is the normalized surface element of $\Sph^2$. 
In particular, $\bx^{\rm ph}=\bx$ for every
$f=a\varphi_{\bx}$ with $a\ne0$.
\end{proposition}
\begin{proof}
Locally write $f=|f|e^{i\alpha}$. Then
\[
 \frac{1}{i\kappa}\frac{\nabla_{\Sph^2}g}{g}
 =\frac1\kappa\nabla_{\Sph^2}\alpha
 =\frac1\kappa\operatorname{Im}
        \left(\frac{\nabla_{\Sph^2}f}{f}\right)
 =:\ba_f.
\]
The field $\ba_f$ is globally defined and tangent, independently of the
local phase branch. Since $|g|=1$, the normal equations are
$(\int P_{\bomega}\,d\mu)\by=\int\ba_f\,d\mu$.
The normal matrix is $2I_3/3$, proving uniqueness and
\eqref{eq:phase-normalized-estimator}. For a pure atom,
$\ba_f=P_{\bomega}\bx$.
\end{proof}

The estimate is unchanged by multiplication of $f$ by any nonzero complex
constant. It can be computed from $\nabla f/f$, without constructing an
unwrapped phase provided that $f$ stays away from zero.  

\subsection{Pointwise dominance and phase-gradient recovery}
\label{sec:hybrid-pointwise-dominance}

The following theorem bounds the phase error when each demixed
function is uniformly dominated by a different target atom.
The assumption is pointwise and can also be imposed on a finite
harmonic reconstruction.

\begin{theorem}
\label{thm:hybrid-pointwise-dominance}
Suppose that the demixed functions $f_j\in C^1(\Sph^2;\C)$ admit
representations
\[
 f_j(\bomega)=a_j\varphi_{x_{\pi(j)}}(\bomega)
                    \bigl(1+\xi_j(\bomega)\bigr),
 \qquad j=1,\ldots,s,
\]
where $\pi$ is a permutation of $\{1,\ldots,s\}$, $a_j\ne0$, and
\[
 \xi_j\in C^1(\Sph^2;\C),
 \qquad \|\xi_j\|_{L^\infty(\Sph^2)}\le\rho_j<1.
\]
Then each $f_j$ is nonvanishing, and the phase-gradient estimate
\begin{equation}
 \bx_j^{\mathrm{ph}}
 =\frac{3}{2\kappa}\int_{\Sph^2}
  \operatorname{Im}\left(\frac{\nabla_{\Sph^2}f_j}{f_j}\right)d\mu(\bomega),
 \label{eq:hybrid-phase-estimate}
\end{equation}
where $d\mu$ is the normalized surface element of $\Sph^2$, 
satisfies
\begin{equation}
 \left|\bx_j^{\mathrm{ph}}-\bx_{\pi(j)}\right|
 \le\frac{3}{2\kappa}\arcsin\rho_j.
 \label{eq:hybrid-pointwise-bound}
\end{equation}
Consequently, with $\rho_* =\max_j\rho_j$ and with $X^{\mathrm{ph}}$
understood as the indexed collection of the $s$ estimates,
\begin{equation}
 E_\infty(X^{\mathrm{ph}},X)
 :=\min_{\sigma}\max_j
       |\bx_j^{\mathrm{ph}}-\bx_{\sigma(j)}|
 \le\frac{3}{2\kappa}\arcsin\rho_*,
 \label{eq:hybrid-bottleneck-bound}
\end{equation}
where $\sigma$ runs through  the permutation group on $s$ indices.
\end{theorem}

\begin{proof}
Fix $j$, and abbreviate $\bx=\bx_{\pi(j)}$, $a=a_j$, $\xi=\xi_j$, and
$\rho=\rho_j$. Since $|\varphi_\bx|=1$,
\[
 |f_j(\bomega)|=|a|\,|1+\xi(\bomega)|
 \ge |a|(1-\rho)>0.
\]
Thus the logarithmic derivative in \eqref{eq:hybrid-phase-estimate} 
 is well defined.

The values of $1+\xi$ lie in the closed disk
\[
 \{z\in\C:|z-1|\le\rho\},
\]
which is contained in the open right half-plane. Hence its principal
phase
\[
 \theta(\bomega)=\arg(1+\xi(\bomega))
\]
is a globally defined real $C^1$ function. The angular extent of this
disk, viewed from the origin, gives
\begin{equation}
 \|\theta\|_{L^\infty(\Sph^2)}\le\arcsin\rho.
 \label{eq:hybrid-residual-phase-bound}
\end{equation}

 Logarithmic differentiation yields
\[
 \frac{\nabla_{\Sph^2}f_j}{f_j}
 =i\kappa P_\bomega \bx+
       \frac{\nabla_{\Sph^2}\xi}{1+\xi}.
\]
Taking imaginary parts,
\[
 \operatorname{Im}\left(\frac{\nabla_{\Sph^2}f_j}{f_j}\right)
 =\kappa P_\bomega \bx+\nabla_{\Sph^2}\theta.
\]
Because $\int_{\Sph^2}P_\bomega\,d\mu=2I_3/3$, the coordinate error has
the exact representation
\begin{equation}
 \bx_j^{\mathrm{ph}}-\bx
 =\frac{3}{2\kappa}\int_{\Sph^2}\nabla_{\Sph^2}\theta\,d\mu.
 \label{eq:hybrid-error-representation}
\end{equation}

Let $\bv\in\R^3$ be any unit vector. The spherical gradient is tangential,
and
\[
 \operatorname{div}_{\Sph^2}(P_\bomega \bv)=-2\,\bomega\cdot \bv.
\]
Integration by parts on the sphere therefore gives
\[
 \begin{aligned}
 \bv\cdot\int_{\Sph^2}\nabla_{\Sph^2}\theta\,d\mu
 &=\int_{\Sph^2}(P_\bomega \bv)\cdot\nabla_{\Sph^2}\theta\,d\mu\\
 &=2\int_{\Sph^2}(\bomega\cdot \bv)\theta(\bomega)\,d\mu.
 \end{aligned}
\]
It follows that
\[
 \begin{aligned}
 |\bv\cdot(\bx_j^{\mathrm{ph}}-\bx)|
 &\le\frac{3}{\kappa}\|\theta\|_\infty
                \int_{\Sph^2}|\bomega\cdot \bv|\,d\mu\\
 &=\frac{3}{2\kappa}\|\theta\|_\infty\\
 &\le\frac{3}{2\kappa}\arcsin\rho,
 \end{aligned}
\]
where rotational symmetry gives
$\int_{\Sph^2}|\bomega\cdot \bv|\,d\mu=1/2$.
Taking the supremum over unit vectors $v$ proves
\eqref{eq:hybrid-pointwise-bound}.

The permutation $\pi$ supplies a one-to-one matching between the
estimated coordinates and the targets, proving
\eqref{eq:hybrid-bottleneck-bound}. Finally,
\[
 \frac{\nabla_{\Sph^2}(cf_j)}{cf_j}
 =\frac{\nabla_{\Sph^2}f_j}{f_j},\qquad c\ne0,
\]
which proves invariance under nonzero constant rescaling.
\end{proof}

\begin{corollary}
\label{cor:hybrid-coefficient-dominance}
Suppose
\[
 f_j=\sum_{\bz\in\mathcal Z}a_{\bz j}\varphi_\bz
\]
for a finite set $\mathcal Z\subset\R^3$ containing $X$. If a permutation
$\pi$ satisfies
\[
 a_{\bx_{\pi(j)},j}\ne0,
 \qquad
 r_j:=\frac{\displaystyle
             \sum_{\bz\in\mathcal Z\setminus\{\bx_{\pi(j)}\}}|a_{\bz j}|}
            {|a_{\bx_{\pi(j)},j}|}<1,
\]
then
\[
 |\bx_j^{\mathrm{ph}}-\bx_{\pi(j)}|
 \le\frac{3}{2\kappa}\arcsin r_j.
\]
\end{corollary}
\begin{proof}
Factor out the target component:
\[
 \begin{aligned}
 f_j&=a_{\bx_{\pi(j)},j}\varphi_{\bx_{\pi(j)}}(1+\xi_j),\\
 \xi_j(\bomega)
 &=\sum_{\bz\in\mathcal Z\setminus\{\bx_{\pi(j)}\}}
       \frac{a_{\bz, j}}{a_{\bx_{\pi(j)},j}}
       e^{i\kappa\bomega\cdot(\bz-\bx_{\pi(j)})}.
 \end{aligned}
\]
Since every exponential has modulus one,
$\|\xi_j\|_\infty\le r_j$.
Theorem~\ref{thm:hybrid-pointwise-dominance} applies.
\end{proof}

\begin{corollary}
\label{cor:hybrid-reconstruction-error}
Under the hypotheses of Theorem~\ref{thm:hybrid-pointwise-dominance},
let $\widetilde f_j\in C^1(\Sph^2;\C)$ be a reconstructed demixed
function, and suppose
\[
 \|\widetilde f_j-f_j\|_\infty\le\varepsilon_j,
 \qquad
 \widetilde\rho_j:=\rho_j+\frac{\varepsilon_j}{|a_j|}<1.
\]
Then $\widetilde f_j$ is nonvanishing, and its continuously integrated
phase-gradient estimate satisfies
\begin{equation}
 |\widetilde \bx_j^{\mathrm{ph}}-\bx_{\pi(j)}|
 \le\frac{3}{2\kappa}\arcsin\widetilde\rho_j.
 \label{eq:hybrid-reconstruction-bound}
\end{equation}
In particular, this applies to a finite harmonic reconstruction
$\widetilde f_j=f_{j,L}$ whenever its uniform reconstruction error
obeys the stated condition.
\end{corollary}
\begin{proof}
Write
\[
 \widetilde f_j=a_j\varphi_{x_{\pi(j)}}(1+\widetilde\xi_j),
 \qquad
 \widetilde\xi_j=\xi_j+
       \frac{\widetilde f_j-f_j}{a_j\varphi_{\bx_{\pi(j)}}}.
\]
Since $|\varphi_{\bx_{\pi(j)}}|=1$,
\[
 \|\widetilde\xi_j\|_\infty
 \le\rho_j+\frac{\varepsilon_j}{|a_j|}
 =\widetilde\rho_j<1.
\]
Apply Theorem~\ref{thm:hybrid-pointwise-dominance} to
$\widetilde f_j$.
\end{proof}
The error bound depends on the size of the relative remainder, rather
than on the size of its angular derivative. Pointwise dominance confines
the residual phase to a bounded interval, and integration by parts
converts its averaged gradient into an integral of that bounded phase.

\subsection{Two-atom example}
For a two-atom mixture, the phase error admits a sharper bound. Assume $A\ne0$ and $|\beta|<1$, where $\beta=B/A$ is
complex, and retain $\bd=\bz-\bx$, $r=|\bd|>0$. The factorization
$f=A\varphi_{\bx}(1+\beta e^{i\kappa\bomega\cdot\bd})$ and the
uniformly convergent geometric series give
\[
 \ba_f=P_{\bomega}\bx+
 \operatorname{Re}\sum_{n\ge1}(-1)^{n-1}\beta^n
          e^{in\kappa\bomega\cdot\bd}P_{\bomega}\bd.
\]
Using the kernel \eqref{eq:tangential-kernel}, whose longitudinal eigenvalue
is $\tau_\parallel(t)=-2\sinc'(t)/t$, yields the exact expression
\begin{equation}
 \bx^{\rm ph}-\bx
 =\frac32\sum_{n\ge1}(-1)^{n-1}
       \operatorname{Re}(\beta^n)T_\kappa(n\bd)\bd
 =\frac32\sum_{n\ge1}(-1)^{n-1}
       \operatorname{Re}(\beta^n)
           \tau_\parallel(n\kappa r)\bd.
 \label{eq:phase-two-atom-series}
\end{equation}
Every term contains a decaying oscillatory kernel: there is no surviving
energy-weighted displacement proportional to $|B|^2\bd$ as in the original
estimate. Since $|\tau_\parallel(t)|\le2(t^{-2}+t^{-3})$ for $t>0$,
\begin{equation}
 |\bx^{\rm ph}-\bx|
 \le\frac{3|\beta|}{1-|\beta|}
       \left(\frac1{\kappa^2r}+\frac1{\kappa^3r^2}\right).
 \label{eq:phase-two-atom-bound}
\end{equation}
Thus the error is $O(\kappa^{-2})$ for fixed distinct locations and
$|\beta|$ bounded away from one. The estimate is not uniform as the
locations coalesce or the amplitudes approach equal magnitude.
In the exact equal-angle construction, each fixed
$\eta<1/\sqrt2$ gives $A\to\sqrt{1-\eta^2}$ and $B\to\eta$, so
$|B/A|<1$ for sufficiently large $\kappa$. The phase-normalized error
then tends to zero at this rate, whereas the original least-squares
estimate has the limit \eqref{eq:cmp-generator-floor}.

For multiple targets, the scalar estimates apply after demixing.
If each column satisfies \cref{thm:hybrid-pointwise-dominance} with a
different target, the error bound gives a one-to-one matching.
Deriving this dominance condition from the empirical subspace and the
conditioning of the selected pencil remains a separate problem.
Amplitude normalization must follow demixing: applying it to arbitrary
signal-frame columns generally destroys their linear mixture structure.
The experiment in \cref{sec:phase-readout-experiment} uses the existing
generator eigenbasis for this purpose.

\end{document}